\documentclass[12pt]{article}
\usepackage{amsfonts, graphicx,amssymb,amsmath,color}
\usepackage[dvipdfm]{hyperref}
\hypersetup{citecolor=blue,colorlinks=true}
\usepackage[authoryear,comma]{natbib}

\newtheorem{theorem}{Theorem}

\newtheorem{corollary}[theorem]{Corollary}

\newtheorem{remark}[theorem]{Remark}

\newenvironment{proof}[1][Proof]{\noindent\textbf{#1.} }{\ \rule{0.5em}{0.5em}}

\newcommand{\be}{\begin{equation}}
\newcommand{\ee}{\end{equation}}
\newcommand{\ba}{\begin{eqnarray}}
\newcommand{\ea}{\end{eqnarray}}
\newcommand{\bee}{\begin{equation*}}
\newcommand{\eee}{\end{equation*}}
\newcommand{\baa}{\begin{eqnarray*}}
\newcommand{\eaa}{\end{eqnarray*}}
\newcommand{\nn}{\nonumber\\}

\begin{document}

\title{\textbf{Testing Composite Hypothesis based on the Density Power Divergence }}
\author{Basu, A.$^{1}$; Mandal, A.$^{2}$; Martin, N.$^{3}$ and Pardo, L.$%
^{4} $ \\
$^{1}${\small Indian Statistical Institute, Kolkata 700108, India}\\
$^{2}${\small Department of Statistics, University of Pittsburgh, Pittsburgh  15260, USA}\\
$^{3}${\small Department of Statistics and O.R. II, Complutense University of Madrid, 28003 Madrid, Spain}\\
$^{4}${\small Department of Statistics and O.R. I, Complutense University of
Madrid, 28040 Madrid, Spain} }
\date{\today}
\maketitle

%
%
%
%
%

\begin{abstract}
In any parametric inference problem, the robustness of the procedure is a
real concern. A procedure which retains a high degree of efficiency under
the model and simultaneously provides stable inference under data
contamination is preferable in any practical situation over another
procedure which achieves its efficiency at the cost of robustness or vice
versa. The density power divergence family of  \cite{MR1665873} provides a
flexible class of divergences where the adjustment between efficiency and
robustness is controlled by a single parameter $\beta$.  In this paper we consider general tests of 
parametric hypotheses based on the density power divergence. We establish the asymptotic null distribution of the test statistic and explore
its asymptotic power function. Numerical
results illustrate the performance of the theory developed.

\end{abstract}


%

\noindent\underline{\textbf{keywords and phrases}}: density power
divergence, linear combination of chi-squares, robustness, tests of
hypotheses.

\section{Introduction}

\label{SEC:intro}

Hypothesis testing is one of the fundamental paradigms of statistical
inference. The likelihood ratio test is a key component of the classical
theory of hypothesis testing; however, this test is known to be notoriously
nonrobust under model misspecification and the presence of outliers. 
Many density based minimum distance procedures have been observed to have strong robustness
properties in estimation and testing together with high efficiency, eg.,  \cite{MR2183173} and \cite{MR2830561}. 
Among the available robust tests in the literature, those based on the
class of disparities (\citealp{MR999667} and \citealp{MR1292557}) are known to perform well in practical situations and have many
theoretical advantages. However the effectiveness of these procedures in
continuous models is tempered by the fact that it is necessary to construct
a continuous density estimate of the data generating density as an
intermediate step. The procedure thus becomes substantially more
complicated and loses a part of its appeal. In contrast, none of the density
power divergences require any density estimation to implement their
minimization routines.
 \cite{MR3011625} considered parametric hypothesis testing based on the  
density power divergence
for simple null hypotheses. In this paper we extend, in a nontrivial way, the problem for composite null hypotheses in general populations. To do that we have introduced the minimum density power divergence estimator restricted to a general null hypothesis, i.e. the restricted minimum density power divergence estimator. In order to derive the asymptotic distribution of the new family of test statistics proposed in this paper for testing composite null hypotheses, we need the asymptotic distribution of the restricted minimum density power divergence estimator. Thus the theoretical results presented in this paper require a fresh approach and represent a non-trivial generalization of the \cite{MR3011625} paper.

Let $\left\{ P_{\boldsymbol{\theta }}:\boldsymbol{\theta }\in \Theta
\right\} $ be some identifiable parametric family of probability measures
on a measurable space $(\mathcal{X}$,$\mathcal{A)}$ with an open parameter space $\Theta
\subset {\mathbb{R}}^{p},$ $p\geq 1.$ Measures $P_{\boldsymbol{\theta }}$
are assumed to be described by densities $f_{\theta }=dP_{\theta }/d\mu $
absolutely continuous with respect to a dominating $\sigma $-finite measure $%
\mu $ on $\mathcal{X}$ $.$ Let $X_{1},...,X_{n}$ be a random sample from a
density belonging to the family $\left\{ f_{\boldsymbol{\theta }}:\boldsymbol{\theta 
}\in \Theta \right\} $, where the support of the random variables is independent of the parameter $\boldsymbol{\theta}$. Consider a general null
hypothesis of interest which restricts the parameter to a proper subset $%
\Theta _{0}$ of $\Theta $, i.e.%
\begin{equation}
H_{0}:\boldsymbol{\theta }\in \Theta _{0}~\text{against}~H_{1}:\boldsymbol{%
\theta }\notin \Theta _{0}.  \label{1}
\end{equation}%
In many practical hypothesis testing problems, the restricted parameter
space $\Theta _{0}$ is defined by a set of $r<p$ restrictions of the form
\begin{equation}
\boldsymbol{g}(\boldsymbol{\theta )=0}_{r}  \label{0}
\end{equation}%
on $\Theta $, where $\boldsymbol{g}:\mathbb{R}^{p}\rightarrow \mathbb{R}^{r}$
is a vector-valued function such that the $p\times r$ matrix
\begin{equation}
\mathbf{G}\left( \boldsymbol{\theta }\right) =\frac{\partial \boldsymbol{g}%
^{T}(\boldsymbol{\theta )}}{\partial \boldsymbol{\theta }}  \label{0.0}
\end{equation}%
exists and is continuous in $\boldsymbol{\theta }$ and rank$\left( \mathbf{G}%
\left( \boldsymbol{\theta }\right) \right) =r$. Here $\boldsymbol{0}_r$ denotes the null vector of dimension $r$, and the superscript $T$ in the
above represents the transpose of the matrix.

In general, however, there are no uniformly most powerful tests for solving the
class of problems formulated in (\ref{1}). The canonical approaches for
problems like these include the likelihood ratio test statistic, the Wald
test statistic and the Rao test statistic; see, for instance,  \cite{MR0381045}.
The tests based on disparities (or divergences), already mentioned earlier,
also provide attractive theoretical alternatives for performing the above
tests.

In this paper we will solve the hypothesis testing problem presented in (\ref%
{1}) using the family of density power divergences. Let $\mathcal{G}$ denote
the set of all distributions having densities with respect to the dominating
measure. Given any two densities $h$ and $f$ in $\mathcal{G}$, the density power divergence
between them is defined, as the function of a nonnegative tuning parameter $%
\beta$, as 
\begin{equation}
d_{\beta}(h,f)=\left\{
\begin{array}{ll}
\int\left\{ f^{1+\beta}(x)-\left( 1+\frac{1}{\beta}\right) f^{\beta }(x)h(x)+%
\frac{1}{\beta}h^{1+\beta}(x)\right\} dx, & \text{for}\mathrm{~}\beta>0, \\%
[2ex]
\int h(x)\log\left( \displaystyle\frac{h(x)}{f(x)}\right) dx, & \text{for}%
\mathrm{~}\beta=0.%
\end{array}
\right.  \label{Uno.1}
\end{equation}
The case corresponding to $\beta= 0$ may be derived from the general case by
taking the continuous limit as $\beta \rightarrow 0$, and in this case $%
d_0(h, f)$ is the classical Kullback-Leibler divergence. The quantities
defined in equation (\ref{Uno.1}) are genuine divergences in the sense $%
d_{\beta }(h,f)\geq0$ for all $h,f\in\mathcal{G}$ and all $\beta \geq 0$,
and $d_{\beta}(h,f)$ is equal to zero if and only if the densities $h$ and $%
f $ are identically equal.

In Section  \ref{sec:dpd} we introduce the restricted minimum density power divergence
estimator (RMDPDE); we also study its asymptotic distribution and its
relation with the minimum density power divergence
estimator (MDPDE) in this section. The new family of test statistics
and their asymptotic distributions are presented in Section  \ref{sec:test}. In Section  \ref{SEC:likelihood_ratio}
we describe the relation of the proposed test with the likelihood ratio test
for the normal model, and in Section \ref{sec:weibull} we have considered testing hypotheses for the Weibull model. Numerical results including real data examples are
presented in Section  \ref{sec:simulation}. The problem of tuning parameter selection is taken up in Section \ref{tuning}.  Some concluding remarks are given in Section  \ref{SEC:concluding}.

In the rest of the paper, we will frequently use the standard assumptions of asymptotic inference as given by Assumptions A, B, C and D of  \citet[p. 429]{MR702834}. We will
refer to them as the Lehmann conditions. Some of the proofs will also require the conditions D1--D5 of \citet[p. 304]{MR2830561} which we will refer to as Basu et al. conditions. In order to avoid arresting the flow of the paper, these conditions have been presented in the Appendix.


\section{Restricted Minimum Density Power Divergence Estimator}\label{sec:dpd}


We consider the parametric model of densities $\{f_{\boldsymbol \theta}:
{\boldsymbol \theta}\in \Theta \subset {\mathbb{R}}^{p}\}$; suppose that
we are interested in the estimation of ${{\boldsymbol{\theta }}}$. Let $H$
represent the distribution function corresponding to the density $h$. The
minimum density power divergence functional $T_{\beta }(H)$ at $H$ is
defined by the requirement $d_{\beta }(h,f_{T_{\beta }(H)})=\min_{{{%
\boldsymbol{\theta }}}\in \Theta }d_{\beta }(h,f_{{\boldsymbol{\theta }}})$.
Clearly the term $\int h^{1+\beta }(x)dx$ in (\ref{Uno.1}) has no role in
the minimization of $d_{\beta }(h,f_{{\boldsymbol{\theta }}})$ over ${{%
\boldsymbol{\theta }}}\in \Theta $. Thus the essential objective function to
be minimized in the computation of the minimum density power divergence
functional $T_{\beta }(H)$ reduces to
\begin{equation}
\int \left\{ f_{{\boldsymbol{\theta }}}^{1+\beta }(x)-\left( 1+\frac{1}{%
\beta }\right) f_{{\boldsymbol{\theta }}}^{\beta }(x)h(x)\right\} dx=\int f_{%
{\boldsymbol{\theta }}}^{1+\beta }(x)dx-\left( 1+\frac{1}{\beta }\right)
\int f_{{\boldsymbol{\theta }}}^{\beta }(x)dH(x).  \label{B1}
\end{equation}%
Notice that in the above objective function the density $h$ appears only as
a linear term (unlike, say, the computation of the of the minimum Hellinger
distance functional where the square root of the density $h$ is the relevant
quantity). Thus given a random sample $X_{1},\ldots ,X_{n}$ from the
distribution $H$ we can approximate the above objective function by
replacing $H$ with its empirical distribution function $H_{n}$. For a given tuning
parameter $\beta $, therefore, the MDPDE $\widehat{\boldsymbol{\theta }}%
_{\beta }$ of ${\boldsymbol{\theta }}$ can be obtained by minimizing
\begin{align}
\int f_{{\boldsymbol{\theta }}}^{1+\beta }(x)dx-\left( 1+\frac{1}{\beta }%
\right) \int f_{{\boldsymbol{\theta }}}^{\beta }(x)dH_{n}(x)& =\int f_{{%
\boldsymbol{\theta }}}^{1+\beta }(x)dx-\left( 1+\frac{1}{\beta }\right)
\frac{1}{n}\sum_{i=1}^{n}f_{{\boldsymbol{\theta }}}^{\beta }(X_{i})
\notag \\
& =\frac{1}{n}\sum_{i=1}^{n}V_{\boldsymbol{\theta }}(X_{i})  \label{B2}
\end{align}%
over ${{\boldsymbol{\theta }}}\in \Theta $, where $V_{\boldsymbol{\theta }%
}(x)=\int f_{{\boldsymbol{\theta }}}^{1+\beta }(y)dy-\left( 1+\frac{1}{\beta
}\right) f_{{\boldsymbol{\theta }}}^{\beta }(x)$. In the special case $\beta
=0$, the objective function reduces to $-\frac{1}{n}\sum_{i=1}^{n}\log f_{%
\boldsymbol{\theta }}(X_{i})$; the corresponding minimizer turns out to be
the maximum likelihood estimator (MLE) of $\boldsymbol{\theta }$. The
minimization of the expression in (\ref{B2}) over ${{\boldsymbol{\theta }}}$
does not require the use of a nonparametric density estimate of the true
unknown distribution $H$. Existing theory (e.g. \citealp{de1992smoothing})
shows that in general there is little or no advantage in introducing
smoothing for such functionals which may be empirically estimated using the
empirical distribution function alone, except in very special cases. Using $H_n$ as a substitute for $H$, if possible, is therefore a natural step.

Let $\boldsymbol{u}_{{\boldsymbol{\theta }}}(x)=\frac{\partial }{\partial
\boldsymbol{\theta }}\log f_{{\boldsymbol{\theta }}}(x)$ be the likelihood score
function of the model. Under differentiability of the model the minimization
of the objective function in equation (\ref{B2}) leads to an estimating
equation of the form
\begin{equation}
\frac{1}{n}\sum_{i=1}^{n}\boldsymbol{u}_{{\boldsymbol{\theta }}}(X_{i})f_{{%
\boldsymbol{\theta }}}^{\beta }(X_{i})-\int \boldsymbol{u}_{{\boldsymbol{%
\theta }}}(x)f_{{\boldsymbol{\theta }}}^{1+\beta }(x)dx=\boldsymbol{0}_{p},
\label{2.2}
\end{equation}%
which is an unbiased estimating equation under the model. Since the
corresponding estimating equation weights the score $\boldsymbol{u}_{%
\boldsymbol{\theta }}(X_{i})$ with the power of the density $f_{\boldsymbol{%
\theta }}^{\beta }(X_{i})$, the outlier resistant behavior of the estimator
is intuitively apparent. See  \cite{MR1665873} and  \cite{MR1859416} for
more details.

The functional $T_{\beta }(H)$ is Fisher consistent; it takes the value $%
\boldsymbol{\theta }_{0}$ when the true density $h=f_{\boldsymbol{\theta }%
_{0}}$ is in the model. When it is not, $\boldsymbol{\theta }_{\beta
}^{h}=T_{\beta }(H)$ represents the best fitting parameter. For brevity we
will suppress the $h$ superscript in the notation for $\boldsymbol{\theta
}_{\beta }^{h}$; $f_{\boldsymbol{\theta }_\beta}$ is the model element closest
to the density $h$ in the density power divergence sense corresponding to tuning parameter $\beta$.

Let $h$ be the true data generating density, and ${\boldsymbol{\theta }}_\beta 
=T_{\beta }(H)$ be the best fitting parameter. To set up the notation we
define the quantities
\begin{align}
\boldsymbol{J}_{\beta }(\boldsymbol{\theta })& =\int \boldsymbol{u}_{{%
\boldsymbol{\theta }}}(x)\boldsymbol{u}_{{\boldsymbol{\theta }}}^{T}(x)f_{{%
\boldsymbol{\theta }}}^{1+\beta }(x)dx+\int \{\boldsymbol{I}_{\boldsymbol{%
\theta }}(x)-\beta \boldsymbol{u}_{\boldsymbol{\theta }}(x)\boldsymbol{u}_{\boldsymbol{%
\theta }}^{T}(x)\}\{h(x)-f_{\boldsymbol{\theta }}(x)\}f_{{\boldsymbol{\theta
}}}^{\beta }(x)dx,  \label{2.3} \\
\boldsymbol{K}_{\beta }(\boldsymbol{\theta })& =\int \boldsymbol{u}_{{%
\boldsymbol{\theta }}}(x)\boldsymbol{u}_{{\boldsymbol{\theta }}}^{T}(x)f_{{%
\boldsymbol{\theta }}}^{2\beta }(x)h(x)dx-\boldsymbol{\xi }_{\beta }({{%
\boldsymbol{\theta }}})\boldsymbol{\xi }_{\beta }^{T}({{\boldsymbol{\theta }}%
}),  \label{2.4}
\end{align}%
where $\boldsymbol{\xi }_{\beta }({{\boldsymbol{\theta }}})=\int \boldsymbol{%
u}_{\mathbf{\theta }}(x)f_{\boldsymbol{\theta }}^{\beta }(x)h(x)dx$, and $%
\boldsymbol{I}_{\boldsymbol{\theta }}(x)=-\frac{\partial }{\partial
\boldsymbol{\theta }}\boldsymbol{u}_{\boldsymbol{\theta }}(x)$ is the so
called information function of the model.

The following results, proved in  \cite{MR2830561}, form the basis of our
subsequent developments.

\begin{theorem}
\label{Theorem1} We assume that the Basu et al. conditions are true. Then

\begin{enumerate}
\item[a)] The minimum density power divergence estimating equation (\ref{2.2}%
) has a consistent sequence of roots $\widehat{\boldsymbol{\theta }}_{n,\beta }$  (denoted, hereafter, as $\widehat{\boldsymbol{\theta }}_{\beta
}$), i.e. $\widehat{\boldsymbol{\theta }}_{\beta
}\underset{n\rightarrow \infty }{\overset{\mathcal{P}}%
{\longrightarrow }}\boldsymbol{\theta }_\beta$.

\item[b)] $n^{1/2}(\widehat{\boldsymbol{\theta }}_{\beta }-{\boldsymbol{\theta }}%
_\beta)$ has an asymptotic multivariate normal distribution with (vector) mean
zero and covariance matrix $\boldsymbol{J}^{-1}\boldsymbol{K}\boldsymbol{J}%
^{-1}$, where $\boldsymbol{J}=\boldsymbol{J}_{\beta }({{\boldsymbol{\theta }}%
}_\beta)$, $\boldsymbol{K}=\boldsymbol{K}_{\beta }({{\boldsymbol{\theta }}}%
_\beta)$ are as in (\ref{2.3}) and (%
\ref{2.4}) respectively.
\end{enumerate}

\end{theorem}
The above result is similar, in content and spirit, to those of \cite{MR640163}.
When the true distribution $H$ belongs to the model so that $H=F_{%
\boldsymbol{\theta }}$ for some $\boldsymbol{\theta }\in \Theta $, the
formula for $\boldsymbol{J}$, $\boldsymbol{K}$ and $\boldsymbol{\xi }$ simplify to
\begin{align}
\boldsymbol{J}& =\boldsymbol{J}_{\beta }(\boldsymbol{\theta })=\int
\boldsymbol{u}_{\boldsymbol{\theta }}(x)\boldsymbol{u}_{\boldsymbol{\theta }%
}^{T}(x)f_{\boldsymbol{\theta }}^{1+\beta }(x)dx,  \label{model_variance} \\
\boldsymbol{K}& =\boldsymbol{K}_{\beta }(\boldsymbol{\theta })=\int
\boldsymbol{u}_{\boldsymbol{\theta }}(x)\boldsymbol{u}_{\boldsymbol{\theta }%
}^{T}(x)f_{\boldsymbol{\theta }}^{1+2\beta }(x)dx-\boldsymbol{\xi }\boldsymbol{\xi }^{T},  \label{m_v} \\
\boldsymbol{\xi }& =\boldsymbol{\xi }_{\beta }(\boldsymbol{\theta })=\int
\boldsymbol{u}_{\boldsymbol{\theta }}(x)f_{\boldsymbol{\theta }}^{1+\beta
}(x)dx.  \notag
\end{align}%

The restricted minimum density power divergence functional $T_{\beta
}^{0}(H) $ at $H$, on the other hand, is the value in the parameter space
which satisfies 
\begin{equation*}
d_{\beta }(h,f_{T_{\beta }^{0}(H)})=\min_{{{\boldsymbol{\theta }}}\in \Theta
_{0}}d_{\beta }(h,f_{{\boldsymbol{\theta }}}),
\end{equation*}%
provided such a minimizer exists. When a random sample $X_{1},\ldots ,X_{n}$ is
available from the distribution $H$, the restricted minimum density power
divergence estimator of $\boldsymbol{\theta }$ minimizes (\ref{B2}) subject
to $\boldsymbol{g}(\boldsymbol{\theta )=0}_{r}$. Under this set up we will
determine, in the next theorem, the asymptotic distribution of the
restricted minimum density power divergence estimator (RMDPDE) $\widetilde{%
\boldsymbol{\theta }}_{\beta }$ of $\boldsymbol{\theta}$. 


\begin{theorem} \label{theorem:dpd}
Assume that the Lehmann and Basu et al. conditions hold. 
Suppose that the true distribution belongs to the model, and $\boldsymbol{\theta}_0 \in \boldsymbol{\Theta}_0$
is the true parameter. Then
the minimum density power divergence estimator $\widetilde{\boldsymbol{%
\theta }}_{\beta }$ of $\boldsymbol{\theta }$ obtained under the constraints
$\boldsymbol{g}(\boldsymbol{\theta })=\boldsymbol{0}_{r}$ of the null
hypothesis has the distribution
\begin{equation*}
n^{1/2}(\widetilde{\boldsymbol{\theta }}_{\beta }-\boldsymbol{\theta }_{0})%
\underset{n\rightarrow \infty }{\overset{\mathcal{L}}{\longrightarrow }}%
\mathcal{N}(\boldsymbol{0}_{p},\boldsymbol{\Sigma }_{\beta }(\boldsymbol{%
\theta }_{0})\boldsymbol{)}
\end{equation*}%
where
\begin{equation*}
\boldsymbol{\Sigma }_{\beta }(\boldsymbol{\theta }_{0})=\boldsymbol{P}%
_{\beta }(\boldsymbol{\theta }_{0})\boldsymbol{K}_{\beta }(\boldsymbol{%
\theta }_0)\boldsymbol{P}_{\beta }(\boldsymbol{\theta }_{0}),
\end{equation*}%
\begin{equation}
\boldsymbol{P=P}_{\beta }(\boldsymbol{\theta }_{0})=\boldsymbol{J}_{\beta
}^{-1}(\boldsymbol{\theta }_{0})-\boldsymbol{Q}_{\beta }(\boldsymbol{\theta }%
_{0})\boldsymbol{G}^{T}(\boldsymbol{\theta }_{0}) \boldsymbol{J}_{\beta }^{-1}(\boldsymbol{\theta }_{0}),  \label{P}
\end{equation}%
%
%
%
\begin{equation}
\boldsymbol{Q=Q}_{\beta }(\boldsymbol{\theta }_{0})=\boldsymbol{J}_{\beta
}^{-1}(\boldsymbol{\theta }_{0})\boldsymbol{G}(\boldsymbol{\theta }_{0})%
\left[ \boldsymbol{G}^{T}(\boldsymbol{\theta }_{0})\boldsymbol{J}_{\beta
}^{-1}(\boldsymbol{\theta }_{0})\boldsymbol{G}(\boldsymbol{\theta }_{0})%
\right] ^{-1}. \label{Q}
\end{equation}
and $\boldsymbol{J}_{\beta}(\boldsymbol{\theta }_{0})$ is as defined in (\ref{model_variance}),
evaluated at $\boldsymbol{\theta} = \boldsymbol{\theta}_0$.
\end{theorem}

\begin{proof}
See the Appendix.
\hspace*{\fill}\bigskip
\end{proof}



\section{Testing Parametric Composite Hypotheses using Density Power
Divergence} \label{sec:test}

Suppose $\widehat{\boldsymbol{\theta }}_{\beta }$ is the unconstrained estimator of $\boldsymbol{\theta }$, whereas $\widetilde{\boldsymbol{\theta }}_{\beta }$ is the RMDPDE under the null hypothesis  given in (\ref{1}).  
In this section we will present the family of the density power divergence test statistics (DPDTS) for testing the composite null
hypothesis in (\ref{1}). This family of test statistics has the
expression
\begin{equation}
T_{\boldsymbol{\gamma }}(\widehat{\boldsymbol{\theta }}_{\beta },\widetilde{%
\boldsymbol{\theta }}_{\beta })=2nd_{\gamma }(f_{\widehat{\boldsymbol{\theta
}}_{\beta }},f_{\widetilde{\boldsymbol{\theta }}_{\beta }}) , \label{TRES.1}
\end{equation}%
%
%
%
%
%
%
%
%
%
where $d_{\gamma }(f_{\widehat{\boldsymbol{\theta }}_{\beta }},f_{\widetilde{%
 \boldsymbol{\theta }}_{\boldsymbol{\beta }}})$ is given in (\ref{Uno.1}). 
%
%
%
%
%
%
%
%
%
In the following theorem we present the asymptotic distribution of the
family of DPDTS defined in (\ref{TRES.1}).

\begin{theorem}
\label{Theorem3} Assume that the Lehmann and Basu et al. conditions hold. The asymptotic distribution of $T_{\boldsymbol{\gamma }}(%
\widehat{\boldsymbol{\theta }}_{\beta },\widetilde{\boldsymbol{\theta }}%
_{\beta })$ defined in (\ref{TRES.1}) coincides, under the null hypothesis $%
H_{0}$ given in (\ref{1}), with the distribution of the random variable
\begin{equation*}
{\textstyle\sum\limits_{i=1}^{k}}\lambda _{i}^{\beta ,\gamma }(\boldsymbol{%
\theta }_{0})Z_{i}^{2},
\end{equation*}%
%
%
%
where $Z_{1},\ldots ,Z_{k}$ are independent standard normal variables, $%
\lambda _{1}^{\beta ,\gamma }(\boldsymbol{\theta }_{0}),\ldots ,\lambda _{k}^{\beta ,\gamma }(\boldsymbol{\theta }_{0})$ are the
nonzero eigenvalues of
$\boldsymbol{A}_{\gamma }\left( \boldsymbol{\theta }_{0}\right) \boldsymbol{B%
}_{\beta }\left( \boldsymbol{\theta }_{0}\right) \boldsymbol{K}_{\beta }(%
\boldsymbol{\theta }_{0})\boldsymbol{B}_{\beta }\left( \boldsymbol{\theta }%
_{0}\right) $
and 
\begin{equation}
k = rank\left(\boldsymbol{B}_{\beta }\left(\boldsymbol{\theta}_{0}\right)
\boldsymbol{K}_{\beta }  (\boldsymbol{\theta }_{0})
\boldsymbol{B}_{\beta }\left(\boldsymbol{\theta}_{0}\right)  
\boldsymbol{A}_{\gamma}\left(\boldsymbol{\theta}_{0}\right)
\boldsymbol{B}_{\beta }\left(\boldsymbol{\theta}_{0}\right)
\boldsymbol{K}_{\beta } (\boldsymbol{\theta }_{0}) \boldsymbol{B}_{\beta }\left(\boldsymbol{\theta}_{0}\right)
\right)  . \label{k}
\end{equation}%
%
%
%
The matrices $\boldsymbol{A}_{\gamma }\left( \boldsymbol{\theta }_{0}\right)
$ and $\boldsymbol{B}_{\beta }\left( \boldsymbol{\theta }_{0}\right) $ are
defined by%
\begin{equation}
\boldsymbol{A}_{\gamma }\left( \boldsymbol{\theta }_{0}\right) =\left(
a_{ij}^{\gamma }\left( \boldsymbol{\theta }_{0}\right) \right)
_{i,j=1,...,p}=\left( 1+\gamma \right) \left( \int\nolimits_{\mathcal{X}}f_{%
\boldsymbol{\theta }_{0}}^{\gamma -1}\left( x\right) \frac{\partial f_{%
\boldsymbol{\theta }_{0}}\left( x\right) }{\partial \theta _{j}}\frac{%
\partial f_{\boldsymbol{\theta }_{0}}\left( x\right) }{\partial \theta _{i}}%
dx\right) _{i,j=1,...,p} , \label{AA}
\end{equation}%
%
%
%
%
and
\begin{equation}
\boldsymbol{B}_{\beta }\left( \boldsymbol{\theta }_{0}\right) =\boldsymbol{J}%
_{\beta }^{-1}(\boldsymbol{\theta }_{0})\boldsymbol{G}(\boldsymbol{\theta }%
_{0})\left[ \boldsymbol{G}^{T}(\boldsymbol{\theta }_{0})\boldsymbol{J}%
_{\beta }^{-1}(\boldsymbol{\theta }_{0})\boldsymbol{G}(\boldsymbol{\theta }%
_{0})\right] ^{-1}\boldsymbol{G}^{T}(\boldsymbol{\theta }_{0}) \boldsymbol{J}%
_{\beta }^{-1}(\boldsymbol{\theta }_{0}).  \label{B}
\end{equation}%
In the above $\boldsymbol{\theta }_{0} \in \Theta_0$ represents the true unknown value of $\boldsymbol{\theta }$.
%
%
%
\end{theorem}

\begin{proof}
See the Appendix.
\hspace*{\fill}\bigskip
\end{proof}


\begin{remark} 
 The main point to note in the above proof is that it is by no means a trivial or simple extension of Theorem 1 of \cite{MR3011625}. The proof of the latter theorem simply requires the results involving the unrestricted MDPD estimator which has been very well studied in the literature. In the present scenario, one has to deal with both the restricted and unrestricted MDPD estimators. The random nature of the second argument of the DPD makes the derivations substantially more complicated and entirely different techniques have to be applied to the proof of Theorem \ref{Theorem3} in this paper. The restricted MDPD estimator which we have employed here only has a limited presence in the literature. In some sense Theorem \ref{theorem:dpd} may also be considered to be a part of Theorem \ref{Theorem3}, but here we have presented them separately for pedagogical reasons as well as to keep a clear focus in our presentations. 
\end{remark}

\begin{remark} 
We observe that the ranks of the matrices
$
\boldsymbol{B}_{\beta }\left(
\boldsymbol{\theta }_{0}\right) \boldsymbol{K}_{\beta }(\boldsymbol{\theta }%
_{0})\boldsymbol{B}_{\beta }\left( \boldsymbol{\theta }_{0}\right) $
and $\boldsymbol{B}_{\beta }\left(\boldsymbol{\theta}_{0}\right)
\boldsymbol{K}_{\beta }  (\boldsymbol{\theta }_{0})
\boldsymbol{B}_{\beta }\left(\boldsymbol{\theta}_{0}\right)  
\boldsymbol{A}_{\gamma}\left(\boldsymbol{\theta}_{0}\right)
\boldsymbol{B}_{\beta }\left(\boldsymbol{\theta}_{0}\right)
\boldsymbol{K}_{\beta } (\boldsymbol{\theta }_{0}) \boldsymbol{B}_{\beta }\left(\boldsymbol{\theta}_{0}\right)$ 
are equal. 
Moreover, it can be easily shown that
$rank\left(\boldsymbol{B}_{\beta }\left(
\boldsymbol{\theta }_{0}\right) \boldsymbol{K}_{\beta }(\boldsymbol{\theta }%
_{0})\boldsymbol{B}_{\beta }\left( \boldsymbol{\theta }_{0}\right) \right)
=rank(\boldsymbol{G}(\boldsymbol{\theta }_{0}))=r$. So $k=r$, i.e. there will be exactly $r$ non-zero eigenvalues.
\end{remark}
%
%

\begin{corollary} 
\label{corollary4} For the special case when we test the null
hypothesis $H_{0}:\mu =\mu _{0}$ against $H_{1}:\mu \neq \mu _{0}$ under the $\mathcal{N}(\mu ,\sigma ^{2})$
model with $\sigma ^{2}$ unknown, the matrix $\boldsymbol{A}_{\gamma }\left(
\boldsymbol{\theta }_{0}\right) \boldsymbol{B}_{\beta }\left( \boldsymbol{%
\theta }_{0}\right) \boldsymbol{K}_{\beta }(\boldsymbol{\theta }_{0})%
\boldsymbol{B}_{\beta }\left( \boldsymbol{\theta }_{0}\right) $ has the form
\begin{equation*}
\boldsymbol{A}_{\gamma }\left( \boldsymbol{\theta }_{0}\right) \boldsymbol{B}%
_{\beta }\left( \boldsymbol{\theta }_{0}\right) \boldsymbol{K}_{\beta }(%
\boldsymbol{\theta }_{0})\boldsymbol{B}_{\beta }\left( \boldsymbol{\theta }%
_{0}\right) =\left(
\begin{array}{cc}
\displaystyle\frac{1}{\sigma ^{\gamma }}\frac{(\beta +1)^{3}}{\sqrt{\gamma +1%
}(2\beta +1)^{3/2}(2\pi )^{\gamma /2}} & 0 \\
0 & 0%
\end{array}%
\right) ,
\end{equation*}%
%
%
%
%
%
%
%
%
%
%
so that the  DPDTS $T_{\boldsymbol{\gamma }}(\widehat{%
\boldsymbol{\theta }}_{\beta },\widetilde{\boldsymbol{\theta }}_{\beta
})=2nd_{\gamma }(f_{\widehat{\boldsymbol{\theta }}_{\beta }},f_{\widetilde{%
\boldsymbol{\theta }}_{\beta }})$ has the same asymptotic distribution as
that of $\lambda _{1}\chi ^{2}(1)$, where $\lambda _{1}$ is the only nonzero
eigenvalue of the above matrix (equal to its $(1,1)$th element). In
particular when $\gamma =0$ and $\beta =0$, this eigenvalue becomes one, so
that the DPDTS
\begin{equation}
T_{\gamma =0}(\hat{\boldsymbol{\theta }}_{\beta =0},\tilde{\boldsymbol{%
\theta }}_{\beta =0}) = 2nd_{\gamma=0}\left( f_{\hat{\boldsymbol{\theta }}_{\beta =0}},f_{\tilde{\boldsymbol{%
\theta }}_{\beta =0}}\right)
\end{equation}
has a simple asymptotic $\chi ^{2}(1)$ distribution. We
will revisit this problem again in Section \ref{SEC:likelihood_ratio}.
\end{corollary}

A simple approach to approximate the critical region
of the DPDTS and perform the test could be the following. The $k$ eigenvalues described in Theorem \ref{Theorem3} can be expressed as a function of the parameter $\boldsymbol{\theta }_0$. Under the null they can be consistently estimated by replacing $\widetilde{\boldsymbol{\theta }}_{\beta }$ in place of $\boldsymbol{\theta }_0$. Let $\hat{\lambda}_1, \hat{\lambda}_2, \cdots , \hat{\lambda}_k$ represent the corresponding estimated eigenvalues. Generating independent observations $Z_1, Z_2, \cdots, Z_k$ from the $N(0,1)$ distribution repeatedly, one can estimate the quantiles of the distribution of $\sum_{i=1}^k \hat{\lambda}_i Z_i^2$, where $\hat{\lambda}_i$'s are kept fixed during this exercise. The quantiles are then consistent approximations of the true quantiles of the asymptotic null distribution of the statistic in Theorem \ref{Theorem3};  the experimenter can then perform the test based on the critical values thus obtained. In particular when $k=1$ one can perform the test by comparing $T_{\gamma }(\widehat{\boldsymbol{%
\theta }}_{\beta },\widetilde{\boldsymbol{\theta }}_{\beta })/\hat\lambda_1$ with the appropriate upper quantile of $\chi^2(1)$ distribution.
Tables of the
cumulative distribution of ${\sum\limits_{i=1}^{k}}c_{i}Z_{i}^{2}$ are also available in   \cite{soloman1960distribution}, \cite{johnson1968tables}, \cite{eckler1969survey} and  \cite{MR0152069}, which may be helpful in performing the test. \cite{davies1980algorithm} has proposed an algorithm to calculate the critical region corresponding to a linear combination of $\chi^2$ random variables. 
Several other conservative approximations of the critical value of the
DPDTS are provided in \cite{MR3011625}.

\subsection{The Power Function}

By  Theorem \ref{Theorem3} the null hypothesis should be rejected if $%
T_{\gamma }(\widehat{\boldsymbol{\theta }}_{\beta },\widetilde{\boldsymbol{%
\theta }}_{\beta })\geq c_{\alpha }^{\beta ,\gamma },$ where $c_{\alpha
}^{\beta ,\gamma ,}$ is the quantile of order $(1-\alpha )$ of the
asymptotic distribution of $T_{\gamma }(\widehat{\boldsymbol{\theta }}%
_{\beta },\widetilde{\boldsymbol{\theta }}_{\beta })$ under $H_0$. The following theorem
can be used to approximate the power function. 

 \begin{theorem}  
Suppose Lehmann and Basu et al. conditions are satisfied. Assume that $\boldsymbol{%
\theta }^{}\notin \Theta_0$ is the true value of the parameter such that   $\widehat{\boldsymbol{\theta }}_{\beta }\overset{p}{\underset{%
n\rightarrow \infty }{\longrightarrow }}\boldsymbol{\theta }^{}$  under $H_1$. Suppose  there exists $\boldsymbol{\theta }^*\in \Theta_0$ such that the
RMDPDE $\widetilde{\boldsymbol{\theta }}_{\beta }$ of $\boldsymbol{%
\theta }$ satisfies $\widetilde{\boldsymbol{%
\theta }}_{\beta }\overset{p}{\underset{n\rightarrow \infty }{%
\longrightarrow }}\boldsymbol{\theta }^*$. Further assume that
%
\begin{equation}
n^{1/2}\left( (\widehat{\boldsymbol{\theta }}_{\beta },\widetilde{%
\boldsymbol{\theta }}_{\beta })-\left( \boldsymbol{\theta }^{},%
\boldsymbol{\theta }^*\right) \right)^T \overset{\mathcal{L}}{\underset{%
n\rightarrow \infty }{\longrightarrow }}\mathcal{N}\left( \left(
\begin{array}{l}
\boldsymbol{0}_{p} \\
\boldsymbol{0}_{p}%
\end{array}%
\right) ,\left(
\begin{array}{ll}
\boldsymbol{J}_{\beta }^{-1}(\boldsymbol{\theta }^{})\boldsymbol{K}_{\beta
}(\boldsymbol{\theta }^{})\boldsymbol{J}_{\beta }^{-1}(\boldsymbol{\theta }%
^{}) & {\boldsymbol{A}}_{12}\left(\boldsymbol{\theta}^{}, \boldsymbol{\theta
}^*\right) \\
{\boldsymbol{A}}_{12}\left(\boldsymbol{\theta}^{}, \boldsymbol{\theta
}^*\right)^T & \boldsymbol{\Sigma}\left(\boldsymbol{\theta}^{}, \boldsymbol{\theta
}^*\right)%
\end{array}%
\right) \right) ,
\label{a}
\end{equation}%
%
%
%
%
%
%
%
%
%
%
where ${\boldsymbol{A}}_{12}\left(
\boldsymbol{\theta }^{},\boldsymbol{\theta }^*\right) $ and
$\boldsymbol{\Sigma }\left( \boldsymbol{\theta }^*, \boldsymbol{\theta}^{}\right) $
are appropriate $p\times p$
matrices. Then, under $H_1$, we have the following convergence
\begin{equation*}
n^{1/2}\left( d_{\gamma }(f_{\widehat{\boldsymbol{\theta }}_{\beta }},f_{%
\widetilde{\boldsymbol{\theta }}_{\beta }})-d_{\gamma }(f_{\boldsymbol{%
\theta }^{}},f_{\boldsymbol{\theta }^*})\right) \overset{L}{\underset{%
n\rightarrow \infty }{\longrightarrow }}\mathcal{N}\left( 0,\sigma _{\beta
,\gamma }^{2}\left( \mathbf{\boldsymbol{\theta }}^{},\boldsymbol{\theta
}^*\right) \right),
\end{equation*}%
%
%
%
%
%
%
where
\begin{equation}
\sigma _{\beta ,\gamma }^{2}\left( \mathbf{\boldsymbol{\theta }}^{},%
\boldsymbol{\theta }_0\right) =\boldsymbol{t}^{T}\boldsymbol{J}_{\beta
}^{-1}(\boldsymbol{\theta }^{})\boldsymbol{K}_{\beta }(\boldsymbol{\theta }%
^{})\boldsymbol{J}_{\beta }^{-1}(\boldsymbol{\theta }^{})\boldsymbol{t+}2%
\boldsymbol{t}^{T}\boldsymbol{A}_{12}\left(\boldsymbol{\theta}^{}, \boldsymbol{\theta }^*\right)\boldsymbol{s+s}^{T}\boldsymbol{\Sigma }%
\left(\boldsymbol{\theta}^{}, \boldsymbol{\theta }^*\right) \boldsymbol{s},  \label{P1}
\end{equation}%
%
%
%
%
%
%
and 
\begin{equation*}
\boldsymbol{t=}\left( \frac{\partial d_{\gamma }(f_{\boldsymbol{\theta }%
_{1}},f_{\boldsymbol{\theta }^*})}{\partial \boldsymbol{\theta }_{1}}%
\right) _{\boldsymbol{\theta }_{1}=\mathbf{\boldsymbol{\theta }}^{}}%
\text{ and }\boldsymbol{s=}\left( \frac{\partial d_{\gamma }(f_{\mathbf{%
\boldsymbol{\theta }}^{}},f_{\boldsymbol{\theta }_{2}})}{\partial
\boldsymbol{\theta }_{2}}\right) _{\boldsymbol{\theta }_{2}=\mathbf{%
\boldsymbol{\theta }}^*}.
\end{equation*}%
%
%
%
%
%
%
\end{theorem}

\begin{proof}
The result follows in a straightforward manner by considering a first order
Taylor expansion of $d_{\gamma }(f_{\widehat{\boldsymbol{\theta }}_{\beta
}},f_{\widetilde{\boldsymbol{\theta }}_{\beta }})$, which yields
\begin{equation*}
d_{\gamma }(f_{\widehat{\boldsymbol{\theta }}_{\beta }},f_{\widetilde{%
\boldsymbol{\theta }}_{\beta }})=d_{\gamma }(f_{\boldsymbol{\theta }^{}},f_{\boldsymbol{\theta }^*})+\boldsymbol{t}^{T}(\widehat{\boldsymbol{%
\theta }}_{\beta }-\boldsymbol{\theta }^{})+\boldsymbol{s}^{T}(%
\widetilde{\boldsymbol{\theta }}_{\beta }-\boldsymbol{\theta }^*)+o\left(
\left\Vert \widehat{\boldsymbol{\theta }}_{\beta }-\boldsymbol{\theta }%
^{}\right\Vert +\left\Vert \widetilde{\boldsymbol{\theta }}_{\beta }-%
\boldsymbol{\theta }^*\right\Vert \right) .
\end{equation*}%
%
\end{proof}

\begin{remark} 
On the basis of the previous theorem we  get an approximation of the
power function as
%
%
\begin{eqnarray}
\pi _{n,\alpha }^{\beta ,\gamma }\left( \mathbf{\boldsymbol{\theta }}^{}\right) &=&\text{P}_{\mathbf{\boldsymbol{\theta }}^{}}\left( T_{\gamma
}(\widehat{\boldsymbol{\theta }}_{\beta },\widetilde{\boldsymbol{\theta }}%
_{\beta })\geq c_{\alpha }^{\beta ,\gamma }\right) \nonumber\\
&=&1-\Phi \left( \frac{n^{1/2}}{\sigma _{\beta ,\gamma }\left(
\mathbf{\boldsymbol{\theta }}^{},\boldsymbol{\theta }^*\right) }%
\left( \frac{c_{\alpha }^{\beta ,\gamma }}{2n}-d_{\gamma }(f_{\boldsymbol{%
\theta }^{}},f_{\boldsymbol{\theta }^*})\right) \right) ,  \label{P2}
\end{eqnarray}%
%
%
%
%
%
%
%
%
%
where $\Phi (\cdot)$ is the standard normal distribution function, $c_{\alpha
}^{\beta ,\gamma }$ is the quantile of order $1-\alpha $ of the asymptotic
distribution of $T_{\gamma }(\widehat{\boldsymbol{\theta }}_{\beta },%
\widetilde{\boldsymbol{\theta }}_{\beta })$ under the null hypothesis, and $\sigma _{\beta ,\gamma
}^{2}\left( \mathbf{\boldsymbol{\theta }}^{},\boldsymbol{\theta }%
^*\right) $ is as defined in (\ref{P1}).

If some $\mathbf{\boldsymbol{\theta }}^{}$ $\neq \boldsymbol{\theta }%
^*$ is the true parameter, then the probability of rejecting $H_0$ 
for a fixed  size $\alpha $ tends to one as $n\rightarrow
\infty $. So the test statistic is consistent in the \citeauthor{MR0093863}'s (\citeyear{MR0093863}) sense.

Obtaining the approximate sample size $n$ to guarantee a power of $\pi$ at
a given alternative $\mathbf{\boldsymbol{\theta }}^{}$ is an
interesting application of formula (\ref{P2}). Let $n^{}$ be the
positive root of the equation
\begin{equation*}
\pi =1-\Phi \left( \frac{n^{1/2}}{\sigma _{\beta ,\gamma }\left( \mathbf{%
\boldsymbol{\theta }}^{},\boldsymbol{\theta }^*\right) }\left( \frac{%
c_{\alpha }^{\beta ,\gamma }}{2n}-d_{\gamma }(f_{\boldsymbol{\theta }^{}},f_{\boldsymbol{\theta }^*})\right) \right),
\end{equation*}%
%
%
%
%
%
%
i.e.%
\begin{equation*}
n^{}=\frac{A+B+\sqrt{A(A+2B)}}{2d_{\gamma }(f_{\boldsymbol{\theta }%
^{}},f_{\boldsymbol{\theta }^*})^{2}},
\end{equation*}%
%
%
%
%
%
%
where
\begin{equation*}
A=\sigma _{\beta ,\gamma }^{2}\left( \mathbf{\boldsymbol{\theta }}^{},%
\boldsymbol{\theta }^*\right) \left( \Phi ^{-1}\left( 1-\pi \right)
\right) ^{2},
\end{equation*}%
%
%
%
%
%
%
and $B=c_{\alpha }^{\beta ,\gamma }d_{\gamma }(f_{\boldsymbol{\theta }^{}},f_{\boldsymbol{\theta }^*}).$ Then the required sample size is $n=\left[
n^{}\right] +1,$ where $\left[ \cdot \right] $ is used to denote
\textquotedblleft integer part of\textquotedblright .
\end{remark}

We may also find an alternative approximation of the power of $T_{\gamma }(\widehat{%
\boldsymbol{\theta }}_{\beta },\widetilde{\boldsymbol{\theta }}_{\beta })$
at an alternative close to the null hypothesis. Let $\boldsymbol{\theta }%
_{n}\in \Theta -\Theta _{0}$ be a given sequence of alternatives, and let $\boldsymbol{%
\theta }_{0}$ be the element in $\Theta _{0}$ closest to $\boldsymbol{\theta }%
_{n}$ in the Euclidean distance sense. One possibility to introduce
contiguous alternative hypotheses is to consider a fixed $\boldsymbol{d}\in
\mathbb{R}^{p}$ and to permit $\boldsymbol{\theta }_{n}$ to move towards $%
\boldsymbol{\theta }_{0}$ as $n$ increases in the manner specified by the hypothesis 
\begin{equation}
H_{1,n}:\boldsymbol{\theta }_{n}=\boldsymbol{\theta }_{0}+n^{-1/2}%
\boldsymbol{d}.  \label{P3}
\end{equation}%
%
%
%
%
%
%
%
%
%
%

\begin{theorem}\label{theorem9}
Suppose that the model satisfies the Lehmann and Basu et al. conditions.
Under the contiguous alternative hypotheses $H_{1,n}$ given in (\ref{P3}),
the asymptotic distribution of $T_{\gamma }(\widehat{\boldsymbol{\theta }}%
_{\beta },\widetilde{\boldsymbol{\theta }}_{\beta })$ coincides with the
distribution of
\begin{equation*}
{\sum\limits_{i=1}^{k}}\lambda _{i}^{\beta ,\gamma }(\boldsymbol{\theta }%
_{0})\left( Z_{i}+w_{i}\right) ^{2}+\eta,
\end{equation*}%
%
%
%
%
%
%
%
%
%
%
where $Z_{1},\ldots ,Z_{r}$ are independent standard normal variables, $%
\lambda _{1}^{\beta ,\gamma }(\boldsymbol{\theta }_{0}),\ldots ,\lambda
_{k}^{\beta ,\gamma }(\boldsymbol{\theta }_{0})$ are the positive
eigenvalues of $\boldsymbol{A}_{\gamma }\left( \boldsymbol{\theta }%
_{0}\right) \boldsymbol{B}_{\beta }\left( \boldsymbol{\theta }_{0}\right)
\boldsymbol{K}_{\beta }(\boldsymbol{\theta }_{0})\boldsymbol{B}_{\beta
}\left( \boldsymbol{\theta }_{0}\right) $, and the values of $\boldsymbol{w=}%
\left( w_{1},\ldots ,w_{k}\right) ^{T}$ and $\eta $ are given by
\begin{equation*}
\boldsymbol{w}=\boldsymbol{\Lambda }_{k}^{-1}\boldsymbol{V}^{T}\boldsymbol{S}%
^{T}\boldsymbol{A}_{\gamma }(\boldsymbol{\theta }_{0})\boldsymbol{B(%
\boldsymbol{\theta }_{0})\boldsymbol{J}_{\beta }(\boldsymbol{\theta }_{0})d}%
,~~\eta =\left( \boldsymbol{B}(\boldsymbol{\theta }_{0})\boldsymbol{Jd}%
\right) ^{T}\boldsymbol{A}_{\gamma }(\boldsymbol{\theta }_{0})\boldsymbol{B(%
\boldsymbol{\theta }_{0})\boldsymbol{J}_{\beta }(\boldsymbol{\theta }_{0})d}-%
\boldsymbol{w}^{T}\boldsymbol{\Lambda }_{k}\boldsymbol{w}.
\end{equation*}%
%
%
%
%
%
%
%
%
%
%
Also $\boldsymbol{S}$ is any square root of $\boldsymbol{B}_{\beta }\left(
\boldsymbol{\theta }_{0}\right) \boldsymbol{K}_{\beta }(\boldsymbol{\theta }%
_{0})\boldsymbol{B}_{\beta }\left( \boldsymbol{\theta }_{0}\right) $, $%
\boldsymbol{\Lambda }_{k}\boldsymbol{=}diag\boldsymbol{(}\lambda _{1}^{\beta
,\gamma }(\boldsymbol{\theta }_{0}),\ldots ,\lambda _{k}^{\beta ,\gamma }(%
\boldsymbol{\theta }_{0}))$ and $\boldsymbol{V}$ is the matrix of
corresponding orthonormal eigenvectors.
\end{theorem}

\begin{proof}
See the Appendix.
\hspace*{\fill}\bigskip
\end{proof}

From a practical point of view we will estimate the eigenvalues as well as $\boldsymbol{w}$ and $\eta$ by their consistent estimators.

\section{Normal Case: Connection with the Likelihood Ratio Test}

\label{SEC:likelihood_ratio}

Under the $\mathcal{N}(\mu ,\sigma ^{2})$ model, consider the problem of testing
\begin{equation}
H_{0}:\mu =\mu _{0}\text{ versus }H_{1}:\mu \neq \mu _{0} , \label{A}
\end{equation}%
%
%
%
%
%
%
%
where $\sigma$ is an unknown nuisance parameter. In this case the
unrestricted and null parameter spaces are given by $\Theta =\{(\mu ,\sigma
)^T\in {\mathbb{R}}^{2}|\mu \in {\mathbb{R}},\sigma \in {\mathbb{R}}%
^{+}\}$ and $\Theta _{0}=\{(\mu ,\sigma)^T\in {\mathbb{R}}^{2}|\mu =\mu
_{0},\sigma \in {\mathbb{R}}^{+}\}$ respectively. If we consider the
function $g(\boldsymbol{\theta })=\mu -\mu _{0},$ with $\boldsymbol{\theta }%
=\left( \mu ,\sigma \right) ^{T}$, the null hypothesis $H_{0}$ can be
written as
\begin{equation*}
H_{0}:g(\boldsymbol{\theta })=0,
\end{equation*}%
%
%
%
%
%
%
and we are in the situation considered in (\ref{A}). We can observe that in
our case $G\left( \boldsymbol{\theta }\right) =\left( 1,0\right) ^{T}.$
Based on (\ref{B2}) and taking into account the fact that $f_{\boldsymbol{%
\theta }}(x)$ is the normal density with mean $\mu $ and variance $\sigma
^{2}$, the estimator $\widehat{\boldsymbol{\theta }}_{\beta }=(\widehat{\mu }%
_{\beta },\widehat{\sigma }_{\beta })^{T}$ of $\boldsymbol{\theta }=(\mu
,\sigma )^{T}$ is given by
\begin{equation*}
(\widehat{\mu }%
_{\beta },\widehat{\sigma }_{\beta })^{T} =
\arg \min_{(\mu
,\sigma )^{T} \in \mathbb{R} \times \mathbb{R}^+}
\frac{1}{\sigma ^{\beta }\left( 2\pi \right)
^{\frac{\beta }{2}}}\left( \frac{1}{\left( 1+\beta \right) ^{3/2}} -\frac{1}{n\beta }\sum_{i=1}^{n}\exp \left\{ -%
\frac{\beta }{2}\left( \frac{X_{i}-\mu }{{\sigma }}\right) ^{2}\right\} \right) ,
\end{equation*}%
where $\beta >0$.
%
%
Similarly, the estimator $\widetilde{\boldsymbol{\theta }}_{\beta }=\left(
\mu _{0},\widetilde{\sigma }_{\beta }\right) ^{T}$, when $\mu =\mu _{0}$, will be obtained from
\begin{equation*}
\widetilde{\sigma }_{\beta } = \arg \min_{\sigma \in \mathbb{R}^+}
\frac{1}{\sigma ^{\beta }\left( 2\pi
\right) ^{\frac{\beta }{2}}}\left( \frac{1}{\left( 1+\beta \right) ^{3/2}} -\frac{1}{n\beta }\sum_{i=1}^{n}\exp
\left\{ -\frac{1}{2}\beta \left( \frac{X_{i}-\mu _{0}}{\sigma }\right)
^{2}\right\} \right).
\end{equation*}%
%
%
%
%
Simple calculations yield the expressions
\begin{equation*}
\boldsymbol{\boldsymbol{J}}_{\beta }(\boldsymbol{\theta })=\frac{1}{%
\sqrt{1+\beta }\left( 2\pi \right) ^{\beta /2}\sigma ^{2+\beta }}\left(
\begin{array}{cc}
\frac{1}{1+\beta } & 0 \\
0 & \frac{\beta ^{2}+2}{\left( 1+\beta \right) ^{2}}%
\end{array}%
\right),
\end{equation*}%
and
\begin{equation*}
\boldsymbol{K}_{\beta }(\boldsymbol{\theta })=\frac{1}{\sigma ^{2+2\beta
}\left( 2\pi \right) ^{\beta }}\left( \frac{1}{(1+2\beta )^{3/2}}\left(
\begin{array}{cc}
1 & 0 \\
0 & \frac{4\beta ^{2}+2}{1+2\beta }%
\end{array}%
\right) -\left(
\begin{array}{cc}
0 & 0 \\
0 & \frac{\beta ^{2}}{(1+\beta )^{3}}%
\end{array}%
\right) \right) .
\end{equation*}%
Based on these matrices we get
\begin{equation*}
\boldsymbol{B}_{\beta }\left( \boldsymbol{\theta }\right) =\allowbreak
\left(
\begin{array}{cc}
\sigma ^{\beta +2}\left( \beta +1\right) ^{\frac{3}{2}}\left( 2\pi \right)
^{\beta /2} & 0 \\
0 & 0%
\end{array}%
\right) .
\end{equation*}%
On the other hand
\begin{equation*}
\boldsymbol{A}_{\gamma }\left( \boldsymbol{\theta }\right) =\frac{1}{\left( 2\pi \right)
^{\gamma /2}\sigma ^{2+\gamma }(1+\gamma )^{1/2}}\left(
\begin{array}{cc}
1 & 0 \\
0 & \frac{\gamma ^{2}+2}{\left( 1+\gamma \right) }%
\end{array}%
\right),
\end{equation*}%
%
%
%
%
and
\begin{equation}
\boldsymbol{A}_{\gamma }\left( \boldsymbol{\theta }\right) \boldsymbol{B}_{\beta }\left(
\boldsymbol{\theta }\right) \boldsymbol{K}_{\beta }(\boldsymbol{\theta })%
\boldsymbol{B}_{\beta }\left( \boldsymbol{\theta }\right) =\left(
\begin{array}{cc}
\frac{1}{\sigma ^{\gamma }}\frac{\left( \beta +1\right) ^{3}}{\sqrt{\gamma +1%
}\left( 2\beta +1\right) ^{\frac{3}{2}}}\frac{1}{\left( 2\pi \right) ^{\frac{%
\gamma }{2}}} & 0 \\
0 & 0%
\end{array}%
\right) ,  \label{EQ:normal_eigenvalue}
\end{equation}%
%
%
%
%
%
%
%
which is identical to the matrix presented in Corollary \ref{corollary4}. In order to
apply the results of Theorem \ref{Theorem3} in this connection, we need to
get the expression of $T_{\gamma }(\boldsymbol{\widehat{\theta }}_{\beta },\boldsymbol{\widetilde{%
\theta }}_{\beta })$. As in \cite{MR3011625} we have 
\begin{align*}
T_{\gamma }(\widehat{\boldsymbol{\theta} }_{\beta },\widetilde{\boldsymbol{\theta} }_{\beta })&
=2nd_\gamma(f_{\widehat{\boldsymbol{\theta
}}_{\beta }},f_{\widetilde{\boldsymbol{\theta }}_{\beta }}) \\
& =\frac{2n}{\widetilde{\sigma }_{\beta }^{\gamma }\sqrt{1+\gamma }\left(
2\pi \right) ^{\gamma /2}}-\left( 1+\frac{1}{\gamma }\right) \frac{1}{%
\widetilde{\sigma }^{\gamma -1}\left( \gamma \widehat{\sigma }_{\beta }^{2}+%
\widetilde{\sigma }_{\beta }^{2}\right) ^{1/2}\left( 2\pi \right) ^{{\gamma
/2}}}\exp \left( -\frac{1}{2}\frac{\mu _{0}^{2}}{\left( \frac{\widehat{%
\sigma }_{\beta }}{\sqrt{\gamma }}\right) ^{2}}+\frac{\widehat{\mu }_{\beta
}^{2}}{\widehat{\sigma }_{\beta }^{2}}\right) \\
& \times \exp \left( \frac{1}{2}\frac{\left( \widehat{\sigma }_{\beta }^{2}\mu _{0}+%
\widehat{\mu }_{\beta }\left( \frac{\widetilde{\sigma }_{\beta }}{\sqrt{%
\gamma }}\right) ^{2}\right) ^{2}}{\left( \widehat{\sigma }_{\beta
}^{2}+\left( \frac{\widehat{\sigma }_{\beta }}{\sqrt{\gamma }}\right)
^{2}\right) \left( \frac{\widetilde{\sigma }_{\beta }}{\sqrt{\gamma }}%
\right) ^{2}\widehat{\sigma }_{\beta }^{2}}\right) +\frac{1}{\gamma \widehat{%
\sigma }_{\beta }^{\gamma }\sqrt{1+\gamma }\left( 2\pi \right) ^{\gamma /2}}.
\end{align*}%
%
%
%
%
%
%
%
Using Corollary \ref{corollary4} and the single nonzero eigenvalue of the matrix given in (\ref%
{EQ:normal_eigenvalue}), we then get
\begin{equation}
\frac{\widetilde{\sigma }_{\beta } ^{\gamma } \sqrt{\gamma +1}\left( 2\beta +1\right) ^{3/2} \left( 2\pi \right) ^{\gamma
/2}}{\left( \beta +1\right) ^{3}}
T_{\gamma }\left( \widehat{\boldsymbol{\theta} }_{\beta },\widetilde{\boldsymbol{\theta} }_{\beta
}\right) \underset{n\rightarrow \infty }{\overset{L}{\longrightarrow }}\chi
^{2}(1).  \label{dpd_stat_normal}
\end{equation}%
%
%
%
%
%
%
%

A special case of interest is the situation where $\beta =0$ and $\gamma =0.$
The likelihood ratio test for the problem under study is equivalent to the
ordinary $t$-test and one can determine the exact small sample critical
values for this test. On the other hand the standard asymptotic formulation
of the likelihood ratio test leads to the rejection of the null hypothesis
when $-2\log \lambda (X_{1},X_{2},\ldots ,X_{n})>\chi _\alpha^{2}(1)$,
where 
\begin{equation*}
\lambda (X_{1},X_{2},\ldots ,X_{n})=\frac{\sup_{\theta \in \Theta
_{0}}f_{\theta }(X_{1},X_{2},\ldots ,X_{n})}{\sup_{\theta \in \Theta
}f_{\theta }(X_{1},X_{2},\ldots ,X_{n})}
\end{equation*}%
%
%
is the likelihood ratio, and $\chi _\alpha^{2}(1)$ is the quantile of order $(1 - \alpha)$ for the $\chi^2(1)$ distribution. The
MLE of $\boldsymbol{\theta }$ under the parameter space $\Theta $ is
\begin{equation*}
\widehat{\boldsymbol{\theta }}_{n}=\left( \bar{X},\hat{\sigma}_{n}^{2}=\frac{%
1}{n}\sum_{i=1}^{n}{(X_{i}-\bar{X})^{2}}\right) ^{T},
\end{equation*}%
%
%
while the MLE under $\Theta _{0}$ is
\begin{equation*}
\widetilde{\boldsymbol{\theta }}_{n}=\left( \mu _{0},\tilde{\sigma}_{n}^{2}=%
\frac{1}{n}\sum_{i=1}^{n}{(X_{i}-\mu _{0})^{2}}\right) ^{T}.
\end{equation*}%
%
%
Straightforward calculations show that asymptotically we reject the null
hypothesis when
\begin{equation}
-2\log \lambda (X_{1},X_{2},\ldots ,X_{n})=n\log \left( \frac{{\tilde{\sigma}%
_{n}}^{2}}{{\hat{\sigma}_{n}}^{2}}\right) >\chi _\alpha^{2}(1).
\label{EQ:asymptotic_LRT}
\end{equation}%
%
%
This test may be looked upon as the asymptotic likelihood ratio test, as
opposed to the usual $t$-test which may be regarded as the exact version of
the likelihood ratio test for the normal mean problem with unknown variance.

What is the relation of the test statistic $T_{\gamma }(\widehat{\boldsymbol{%
\theta }}_{\beta },\widetilde{\boldsymbol{\theta }}_{\beta })$ given in (\ref%
{TRES.1}) with the above test statistics? In the following we will
demonstrate that for $\gamma =0$ and $\beta =0$, our test statistic
coincides with the asymptotic likelihood ratio test described in (\ref%
{EQ:asymptotic_LRT}). Note that the density power divergence for the case $%
\gamma =0$ between the densities of two normal distributions with different
means and variances is given by
\begin{equation*}
d_{\gamma =0}(f_{{\boldsymbol{\theta}}_1},f_{{\boldsymbol{\theta}}_2})
=\log {\frac{\sigma _{2}}{\sigma _{1}}}-\frac{1}{2}+\frac{1}{2}\frac{%
\sigma _{1}^{2}}{\sigma _{2}^{2}}+\frac{1}{2\sigma _{2}^{2}}(\mu _{1}-\mu
_{2})^{2}.
\end{equation*}%
%
%
%
%
%
%
Therefore for $\gamma =0$ and $\beta =0$, we get
\begin{eqnarray*}
T_{\gamma =0}(\hat{\boldsymbol{\theta }}_{\beta =0},\tilde{\boldsymbol{%
\theta }}_{\beta =0}) 
=n\left( \log {\frac{\tilde{\sigma}_{n}^{2}}{\hat{\sigma}_{n}^{2}}}-1+%
\frac{\hat{\sigma}_{n}^{2}}{\tilde{\sigma}_{n}^{2}}+\frac{(\bar{X}-\mu
_{0})^{2}}{\tilde{\sigma}_{n}^{2}}\right) .
\end{eqnarray*}%
%
%
%
%
%
%
A routine calculation shows that
\begin{equation*}
\frac{\hat{\sigma}_{n}^{2}}{\tilde{\sigma}_{n}^{2}}+\frac{(\bar{X}-\mu
_{0})^{2}}{\tilde{\sigma}_{n}^{2}}=1,
\end{equation*}%
%
%
%
%
%
%
so that
\begin{equation}
T_{\gamma =0}(\hat{\boldsymbol{\theta }}_{\beta =0},\tilde{\boldsymbol{%
\theta }}_{\beta =0})=n\log \left( {\frac{\tilde{\sigma}_{n}^{2}}{\hat{\sigma%
}_{n}^{2}}}\right) , \label{EQ:DPD_LRT}
\end{equation}%
%
%
%
%
%
%
and by equations (\ref{EQ:asymptotic_LRT}) and (\ref{EQ:DPD_LRT}), the
asymptotic likelihood ratio test statistic is exactly same as the DPDTS for $\gamma =0$ and $\beta =0$. Therefore when we are comparing
the usual $t$-test with the test statistic $T_{\gamma =0}(\widehat{%
\boldsymbol{\theta }}_{\beta =0},\widetilde{\boldsymbol{\theta }}_{\beta
=0}),$ we are comparing an exact likelihood ratio test with an asymptotic
likelihood ratio test.


\section{Testing for the Weibull Distribution} \label{sec:weibull}
While the normal model is the most important model where our methods are useful, it is also important to explore the applicability of the method in other models to demonstrate the general nature of the method. For this purpose we will include numerical results based on the Weibull distribution in our subsequent numerical study, together with the results on the normal model. Here we describe the statistic for the Weibull case. 
The probability density function of $\mathcal{W}(\sigma,p)$, a two parameter Weibull distribution, is given by
\bee
f_{\boldsymbol{\theta }}(x) = \frac{p}{\sigma} \left(\frac{x}{\sigma}\right)^{p-1} \exp\left\{ -\left(\frac{x}{\sigma}\right)^p\right\}, \ x>0,
\eee
where $\boldsymbol{\theta }%
=(\sigma, p)^{T}$, and the parameter space is given by $\Theta =\{(\sigma, p
)|\sigma \in {\mathbb{R}}^{+}, p \in {\mathbb{R}}^{+}\}$. 
We are interested in testing
\begin{equation}
H_{0}:\sigma =\sigma _{0}\text{ versus }H_{1}:\sigma \neq \sigma _{0} , \label{w}
\end{equation}
where $p$ is a nuisance parameter. Let us consider the
function $g(\boldsymbol{\theta })=\sigma -\sigma _{0}$. Then, as in the normal case which was considered in Section \ref{SEC:likelihood_ratio}, the null hypothesis $H_{0}$ can be
written as
\begin{equation*}
H_{0}:g(\boldsymbol{\theta })=0,
\end{equation*}
and  $\boldsymbol{G}\left( \boldsymbol{\theta }\right) =\left( 1,0\right) ^{T}.$

Let us define
\bee
\xi_{\alpha, \beta}( \boldsymbol{\theta }) = \int_0^\infty \left(\frac{x}{\sigma}\right)^{\alpha} f_{\boldsymbol{\theta }}^\beta(x) dx ,
\eee
and
\bee
\eta_{\alpha, \beta, \gamma}( \boldsymbol{\theta }) = \int_0^\infty \left(\frac{x}{\sigma}\right)^{\alpha} \left[ \log \left(\frac{x}{\sigma}\right)\right]^{\beta} f_{\boldsymbol{\theta }}^\gamma(x) dx .
\eee
It can be shown that
\be
\xi_{\alpha, \beta}( \boldsymbol{\theta })
=\left(\frac{p}{\sigma}\right)^{\beta-1} \beta^{-\frac{\beta p - \beta + \alpha +1}{p}} 
         \Gamma\left(\frac{\beta p - \beta + \alpha +1}{p}\right) ,
\label{xi}
\ee
and
\be
\eta_{\alpha, \beta, \gamma}( \boldsymbol{\theta }) =  \sigma \left(\frac{p}{\sigma}\right)^\gamma \int_0^\infty y^{\alpha + \gamma p - \gamma} 
            ( \log y)^\beta \exp( -\gamma y^p) dy  ,            
\label{eta}            
\ee
where $\Gamma(\cdot)$ denotes the gamma function.
Note that $\xi_{\alpha, \gamma}( \boldsymbol{\theta }) = \eta_{\alpha, 0, \gamma}( \boldsymbol{\theta })$. For $\beta \neq 0$ the value of $\eta_{\alpha, \beta, \gamma}( \boldsymbol{\theta })$ is calculated using numerical integration.
Let us define
\bee
\boldsymbol{R}_\gamma( \boldsymbol{\theta }) = \int_0^\infty \boldsymbol{u}_{\boldsymbol{\theta }}(x) \boldsymbol{u}_{\boldsymbol{\theta }}^T(x) f_{\boldsymbol{\theta }}^\gamma(x) dx = \left(\begin{array}{cc}
 r_{11} & r_{12} \\  r_{12} & r_{21}
\end{array}\right),
\eee
where $\boldsymbol{u}_{\boldsymbol{\theta }}(x)$, the score function of the Weibull distribution, is given by
\bee
\boldsymbol{u}_{\boldsymbol{\theta }}(x) = \frac{\partial \log f_{\boldsymbol{\theta }}(x)}{\partial \boldsymbol{\theta}}
= \left(\begin{array}{c}
  -\frac{p}{\sigma} + \frac{p}{\sigma}\left(\frac{x}{\sigma}\right)^{p} \\
  \frac{1}{p} + \log\left(\frac{x}{\sigma}\right) - \left(\frac{x}{\sigma}\right)^p \log \left(\frac{x}{\sigma}\right) 
\end{array}\right).
\eee
Then it can be shown that
\bee
r_{11} = \left(\frac{p}{\sigma}\right)^2 \left\{  \xi_{0, \gamma}( \boldsymbol{\theta })
 - 2 \xi_{p, \gamma}( \boldsymbol{\theta })  +   \xi_{2p, \gamma}( \boldsymbol{\theta }) \right\} ,
\eee
\bee
r_{12} =  \frac{p}{\sigma}   \left\{ -\frac{1}{p} \xi_{0, \gamma}( \boldsymbol{\theta }) - \eta_{0, 1, \gamma}( \boldsymbol{\theta })  + 2 \eta_{p, 1, \gamma}( \boldsymbol{\theta })
+ \frac{1}{p}  \xi_{p, \gamma}( \boldsymbol{\theta })  -   \eta_{2p, 1, \gamma}( \boldsymbol{\theta })      \right\} ,
\eee
and
\baa
r_{22} &=& \frac{1}{p^2} \xi_{0, \gamma}( \boldsymbol{\theta }) + \eta_{0, 2, \gamma}( \boldsymbol{\theta })  +   \eta_{2p, 2, \gamma}( \boldsymbol{\theta })   + \frac{2}{p} \eta_{0, 1, \gamma}( \boldsymbol{\theta })   - 2 \eta_{p, 2, \gamma}( \boldsymbol{\theta })  - \frac{2}{p} \eta_{p, 1, \gamma}( \boldsymbol{\theta }) .
\eaa
Now 
\be
\boldsymbol{J}_\gamma( \boldsymbol{\theta }) = \int_0^\infty \boldsymbol{u}_{\boldsymbol{\theta }}(x) \boldsymbol{u}_{\boldsymbol{\theta }}^T(x) f_{\boldsymbol{\theta }}^{1+\gamma}(x) dx = \boldsymbol{R}_{1+\gamma}( \boldsymbol{\theta }),
\label{J}
\ee
\be
\boldsymbol{K}_\gamma( \boldsymbol{\theta }) = \int_0^\infty \boldsymbol{u}_{\boldsymbol{\theta }}(x) \boldsymbol{u}_{\boldsymbol{\theta }}^T(x) f_{\boldsymbol{\theta }}^{1+2\gamma}(x) dx = \boldsymbol{R}_{1+2\gamma}( \boldsymbol{\theta }),
\label{K}
\ee
and 
\be
\boldsymbol{A}_\gamma( \boldsymbol{\theta }) = (1+\gamma) \int_0^\infty \boldsymbol{u}_{\boldsymbol{\theta }}(x) \boldsymbol{u}_{\boldsymbol{\theta }}^T(x) f_{\boldsymbol{\theta }}^{1+\gamma}(x) dx = (1+\gamma) \boldsymbol{R}_{1+\gamma}( \boldsymbol{\theta }).
\label{A_gama}
\ee
Suppose we have two densities $f_{\boldsymbol{\theta}_1}$ and $f_{\boldsymbol{\theta}_2}$ from Weibull family, where $\boldsymbol{\theta}_1 =(\sigma_1, p_1)^T$ and $\boldsymbol{\theta}_2 =(\sigma_2, p_2)^T$. If $\gamma>0$, then using (\ref{xi}) we get from equation (\ref{Uno.1})
\be
d_{\gamma}(f_{\boldsymbol{\theta}_1},f_{\boldsymbol{\theta}_2}) = \xi_{0, 1+\gamma}(\boldsymbol{\theta}_2) - \left( 1+\frac{1}{\gamma}\right) \psi_\gamma(\boldsymbol{\theta}_1, \boldsymbol{\theta}_2) + \frac{1}{\gamma} \xi_{0, 1+\gamma}(\boldsymbol{\theta}_1),  
\label{weibull_dpd_est}
\ee
where
\bee
\psi_\gamma(\boldsymbol{\theta}_1, \boldsymbol{\theta}_2) = \int f_{\boldsymbol{\theta}_2}^{\gamma }(x)f_{\boldsymbol{\theta}_1}(x) dx .
\eee 
The value of $\psi_\gamma(\boldsymbol{\theta}_1, \boldsymbol{\theta}_2)$ can also be calculated using numerical integration. 
For $\gamma=0$ it can be shown that
\begin{align}
d_{\gamma=0}(f_{\boldsymbol{\theta}_1},& f_{\boldsymbol{\theta}_2}) = \log p_1 - \log \sigma_1 + (p_1-1)\eta_{0, 1,1}(\boldsymbol{\theta}_1) - \xi_{p_1,1}(\boldsymbol{\theta}_1)\nn
&+  \log p_2 - \log \sigma_2 +(p_2-1) \log\left(\frac{\sigma_1}{\sigma_2}\right) + (p_2-1)\eta_{0, 1,1}(\boldsymbol{\theta}_1) - \left(\frac{\sigma_1}{\sigma_2}\right)^{p_2} \xi_{p_2,1}(\boldsymbol{\theta}_1).
\label{weibull_dpd_est0}
\end{align}
Using equations (\ref{J})-(\ref{weibull_dpd_est0}) we calculate the test statistic as well as its asymptotic distribution.

Suppose $\widehat{\sigma }_{\beta }$ and $\widehat{p }_{\beta }$ are the unconstrained estimators of $\sigma$ and $p$ respectively, and $\widetilde{p }_{\beta }$ is the RMDPDE of $p$ under the null hypothesis. 
For $\gamma>0$, the test statistic can be simplified as
\begin{align}
& T_{\gamma }(\widehat{\sigma }_{\beta },\widehat{p }_{\beta },\sigma
_{0},\widetilde{p }_{\beta })  \notag \\
& =\frac{2n\widetilde{r}_{\gamma +1}^{(2,2)}(\widetilde{p }_{\beta })\left[ 
\widetilde{r}_{\gamma +1}^{(1,1)}(\widetilde{p }_{\beta })\widetilde{r}_{\gamma
+1}^{(2,2)}(\widetilde{p }_{\beta })-\left( \widetilde{r}_{\gamma +1}^{(1,2)}(%
\widetilde{p }_{\beta })\right) ^{2}\right] }{(1+\gamma )%
\begin{pmatrix}
-\widetilde{r}_{\gamma +1}^{(2,2)}(\widetilde{p }_{\beta }) & \widetilde{r}%
_{\gamma +1}^{(1,2)}(\widetilde{p }_{\beta })%
\end{pmatrix}%
\begin{pmatrix}
\widetilde{r}_{2\gamma +1}^{(1,1)}(\widetilde{p }_{\beta }) & \widetilde{r}%
_{2\gamma +1}^{(1,2)}(\widetilde{p }_{\beta }) \\ 
\widetilde{r}_{2\gamma +1}^{(1,2)}(\widetilde{p }_{\beta }) & \widetilde{r}%
_{2\gamma +1}^{(2,2)}(\widetilde{p }_{\beta })%
\end{pmatrix}%
\begin{pmatrix}
-\widetilde{r}_{\gamma +1}^{(2,2)}(\widetilde{p }_{\beta }) \\ 
\widetilde{r}_{\gamma +1}^{(1,2)}(\widetilde{p }_{\beta })%
\end{pmatrix}%
}  \notag \\
& \times \left\{ \frac{1}{\gamma }\left( \frac{\widehat{p }_{\beta }\sigma
_{0}}{\widetilde{p }_{\beta }\widehat{\sigma }_{\beta }}\right) ^{\gamma
}\varepsilon _{0,\gamma +1}(\widehat{p }_{\beta })+\varepsilon _{0,\gamma
+1}(\widetilde{p }_{\beta })-\frac{\gamma +1}{\gamma }\frac{1}{\sigma
_{0}^{\gamma (\widetilde{p }_{\beta }-1)}}\frac{\widehat{p }_{\beta }}{%
\widehat{\sigma }_{\beta }^{\widehat{p }_{\beta }}}\overline{I}_{\gamma }(%
\widehat{\sigma }_{\beta },\widehat{p }_{\beta },\sigma _{0},\widehat{p}%
_{\beta })\right\},   \notag \label{eq222}
\end{align}%
where
\begin{equation*}
\widetilde{r}_{\gamma }^{(1,1)}(\widetilde{p }_{\beta })=\varepsilon _{0,\gamma
}(\widetilde{p }_{\beta })-2\varepsilon _{\widetilde{p }_{\beta },\gamma }(%
\widetilde{p }_{\beta })+\varepsilon _{2\widetilde{p }_{\beta },\gamma }(\widehat{p%
}_{\beta }),  \label{eq4a}
\end{equation*}
\begin{equation*}
\widetilde{r}_{\gamma }^{(1,2)}(\widetilde{p }_{\beta })=-\frac{1}{\widehat{p}%
_{\beta }}\varepsilon _{0,\gamma }(\widetilde{p }_{\beta })+\left( \log 
\widetilde{p }_{\beta }+\frac{1}{\widetilde{p }_{\beta }}\right) \varepsilon _{%
\widetilde{p }_{\beta },\gamma }(\widetilde{p }_{\beta })-\log \widehat{p}_{\beta
}\varepsilon _{2\widetilde{p }_{\beta },\gamma }(\widetilde{p }_{\beta })-\kappa
_{0,1,\gamma }(\widetilde{p }_{\beta })+\kappa _{\widetilde{p }_{\beta },1,\gamma
}(\widetilde{p }_{\beta }) , \label{eq4b}
\end{equation*}
\begin{align*}
\widetilde{r}_{\gamma }^{(2,2)}(\widetilde{p }_{\beta })& =\frac{1}{\widehat{p}%
_{\beta }^{2}}\varepsilon _{0,\gamma }(\widetilde{p }_{\beta })-\frac{2}{%
\widetilde{p }_{\beta }}\log \widetilde{p }_{\beta }\varepsilon _{\widehat{p}%
_{\beta },\gamma }(\widetilde{p }_{\beta })+(\log \widehat{p}_{\beta
})^{2}\varepsilon _{2\widetilde{p }_{\beta },\gamma }(\widetilde{p }_{\beta }) 
\notag \\
& +\frac{2}{\widetilde{p }_{\beta }}\kappa _{0,1,\gamma }(\widehat{p}_{\beta
})+\kappa _{0,2,\gamma }(\widetilde{p }_{\beta })-2\log \widehat{p}_{\beta
}\kappa _{\widetilde{p }_{\beta },1,\gamma }(\widetilde{p }_{\beta }) , 
\end{align*}
\begin{equation*}
\overline{I}_{\gamma }(\widehat{\sigma }_{\beta },\widetilde{p}_{\beta
},\sigma _{0},\widetilde{p }_{\beta })=\int_{0}^{\infty }x^{\gamma (\widehat{p}%
_{\beta }-1)+\widehat{p }_{\beta }-1}\exp \left\{ -\gamma \left( \frac{x}{%
\sigma _{0}}\right) ^{\widetilde{p }_{\beta }}-\left( \frac{x}{\widetilde{%
\sigma }_{\beta }}\right) ^{\widehat{p }_{\beta }}\right\} dx  ,\label{eq3}
\end{equation*}%
\begin{align*}
\xi _{\alpha ,\gamma }(\sigma ,p)& =\left( \frac{p}{\sigma }\right) ^{\gamma
-1}\varepsilon _{\alpha ,\gamma }(p),  \notag \\
\varepsilon _{\alpha ,\gamma }(p)& =\gamma ^{\frac{(p-1)\gamma +\alpha +1}{p}%
}\Gamma \left( \frac{(p-1)\gamma +\alpha +1}{p}\right) ,  
\end{align*}%
\begin{align*}
\eta _{\alpha ,\delta ,\gamma }(\sigma _{0},\widetilde{p }_{\beta })& =\widehat{%
p}_{\beta }\left( \frac{\widetilde{p }_{\beta }}{\sigma _{0}}\right) ^{\gamma
-1}\int_{0}^{\infty }(\log y)^{\delta }y^{(\widetilde{p }_{\beta }-1)\gamma
+\alpha }\exp \{-\gamma y^{\widetilde{p }_{\beta }}\}dy  \notag \\
& =\left( \frac{\widetilde{p }_{\beta }}{\sigma _{0}}\right) ^{\gamma -1}\kappa
_{\alpha ,\delta ,\gamma }(\widetilde{p }_{\beta }),  \notag \\
\kappa _{\alpha ,\delta ,\gamma }(\widetilde{p }_{\beta })& =\widehat{p}_{\beta
}\int_{0}^{\infty }(\log y)^{\delta }y^{(\widetilde{p }_{\beta }-1)\gamma
+\alpha }\exp \{-\gamma y^{\widetilde{p }_{\beta }}\}dy  .
\end{align*}%

\section{Numerical Studies}\label{sec:simulation}
In this section we provide some extensive numerical evidence of the performance of the proposed methods, 
demonstrating, in particular, their strong robustness properties. Notice that the test statistic depends
on the data only through the value of the estimator (both unconstrained and constrained), so that the 
robustness of the test would appear to depend directly on the robustness of the estimator. However, it is
still useful to develop actual theoretical robustness properties of the proposed tests. Fortunately 
there is a wealth of material available in this context which makes our work easy. \cite{toma2011dual} and \cite{toma2010robust} have, in general, touched upon the issue of theoretical robustness
 properties of tests. 
They have considered several theoretical measures of robustness in this context. In a more limited, but a more focused setting \cite{ghosh2014robustness} have considered the robustness measures of test statistics based on the family of $S$-divergences which include the DPD as a special case; in particular the influence functions of the tests and the so called level and power influence functions are derived. Taken together, the above references further reinforce the notion that the robustness of these tests are directly dependent on the robustness of the estimators as the influence function of the tests turn out to be directly related to the influence function of the estimators.  The \cite{ghosh2014robustness} paper relates only to the case of the simple null hypothesis; however it is not difficult to intuitively see how the robustness of these tests extend to the case of the composite hypothesis. The theoretical robustness properties of some similar tests have been considered  in the Ph.D. dissertation of \cite{GhoshThesis}. 
On the whole, there is substantial overall indication and evidence of the theoretical robustness properties 
of the tests under study. For the sake of brevity we do not repeat these results here, but concentrate instead on 
the performance of the tests as observed in simulations and actual real data examples. 



\subsection{Real Data Examples}



\subsubsection{Telephone-Fault Data}

\label{SEC:Telephone_example}

We consider the data on telephone line faults presented and analyzed by
 \cite{MR909365}; the data were also analyzed by  \cite{MR999667}. The data are
given in Table \ref{TAB:telephone-line-faults}, and consist of the ordered
differences between the inverse test rates and the inverse control rates in
14 matched pairs of areas. A parametric approach to analyze this  would be to model
these data as a random sample from a normal distribution with mean $\mu$ and
standard deviation $\sigma$. It is obvious that
the first observation of this dataset is a huge outlier with respect to the
normal model, while the remaining 13 observations appear to be reasonable
with respect to the same. 
%

\cite{MR3011625} provided a limited analysis of these data by testing simple null hypotheses under the normal model. They tested null hypothesis about the mean by assuming the variance to be known, and also tested null hypothesis about the variance by assuming the mean to be known. These are contrived situations, and are less meaningful than the more realistic situation where both parameters are unknown. In this paper we consider tests for the normal mean without assuming the scale parameter to be known. 



For the full data, the $t$-test for the null hypothesis $H_0: \mu = 0$
against $H_1: \mu \neq 0$ fails to reject the null due to the presence of
the large outlier (two sided $p$-value is 0.6584); however the robust
Hellinger deviance test \citep{MR999667} comfortably rejects the null (two
sided $p$-value based on the chi-square null distribution is 0.0061), as
does the $t$-test based on the cleaned data after the removal of the large
outlier (two sided $p$-value is 0.0076).

\begin{table}[h]
\caption{Telephone-Fault Data}
\label{TAB:telephone-line-faults}
\tabcolsep=0.15cm
\begin{center}
\begin{tabular}{lcccccccccccccc}
\hline
Pair & 1 & 2 & 3 & 4 & 5 & 6 & 7 & 8 & 9 & 10 & 11 & 12 & 13 & 14 \\
Difference & $-988$ & $-135$ & $-78$ & 3 & 59 & 83 & 93 & 110 & 189 & 197 &
204 & 229 & 289 & 310 \\ \hline
\end{tabular}%
\end{center}
\end{table}

\begin{figure}[tbp]
\centering%
\begin{tabular}{rl}
\includegraphics[height=7.5cm, width=8cm]{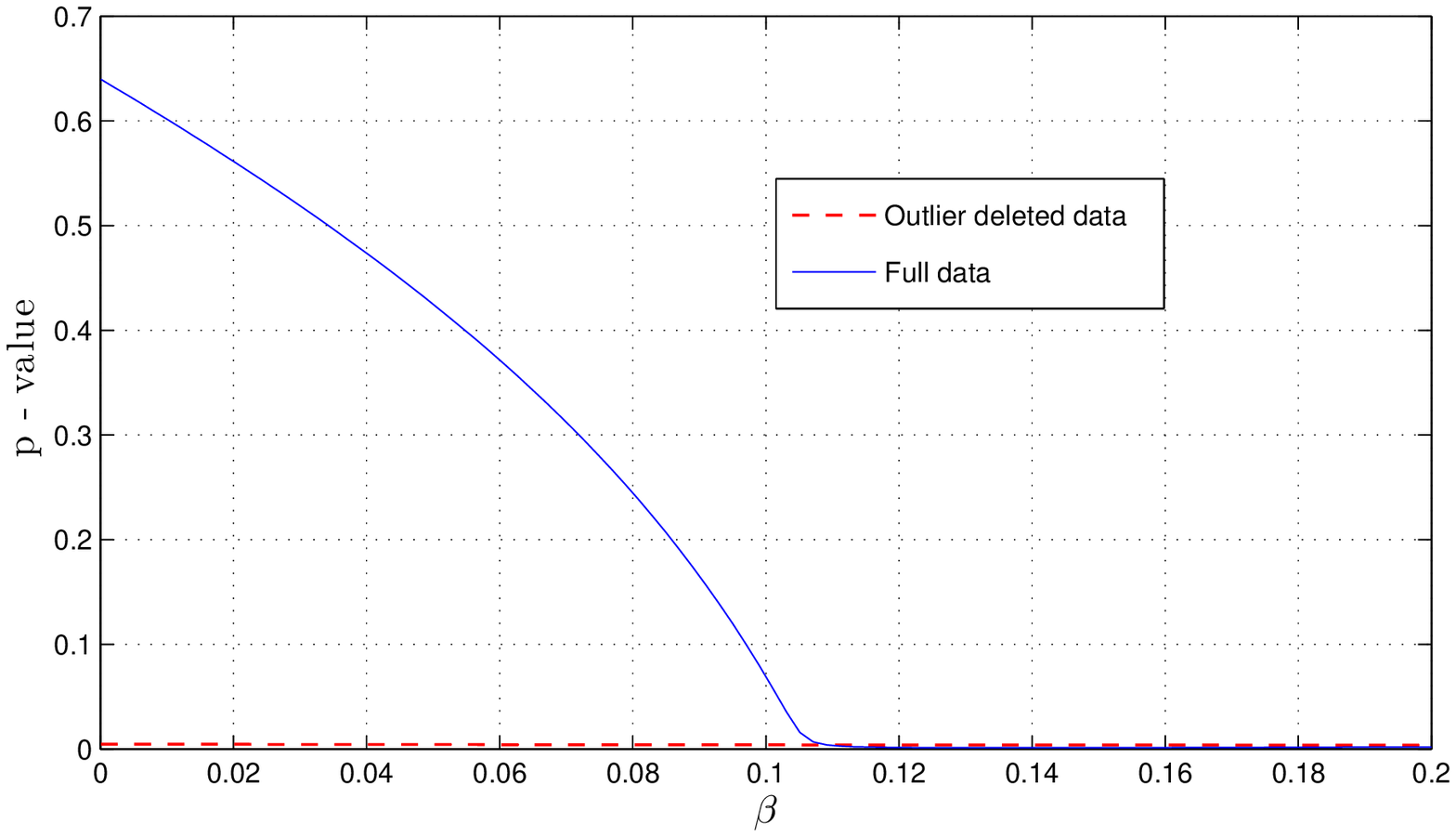}\negthinspace &
\negthinspace \includegraphics[height=7.5cm, width=8cm]{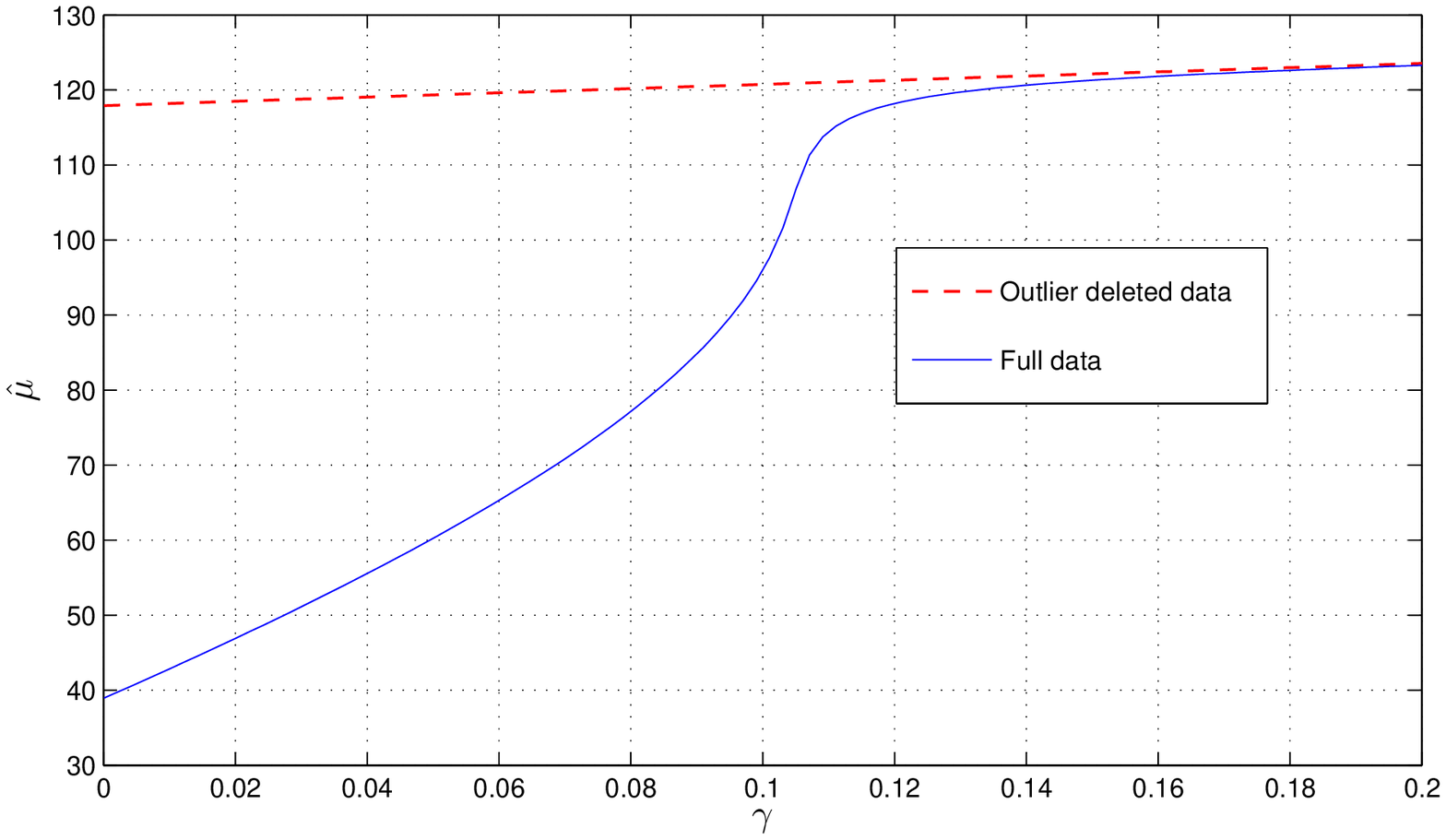} 
\end{tabular}%
\caption{(a) Two sided $p$-values of the density power divergence tests and 
(b) estimates of $\protect\mu$ for different values of $\protect\gamma$ in case of the telephone-fault data.}
\label{fig:Tele_mu0_composite}
\end{figure}

%

Under the normal model, the maximum likelihood estimates of $\mu$ (and $\sigma$) are
highly distorted due to the presence of the large outlier, and as a result
the likelihood ratio test under the normal model fails to reject the null
hypothesis. From the robustness perspective, this is precisely what we will
like to avoid, and here we demonstrate that proper choices of the tuning
parameter within the class of tests developed in this paper achieve this
goal. Here we analyze the performance of the density power divergence tests
with $\beta = \gamma$. Figure \ref{fig:Tele_mu0_composite}(a) represents the $p$%
-values of the test $H_0: \mu = 0$ versus $H_1: \mu \neq 0$  for different values of $\beta$ in a region of
interest. While it is clearly seen that the tests fail to reject the null
hypothesis for these data at very small values of $\beta$, the decision
turns around sharply, as $\beta$ crosses and goes beyond 0.1. On the other hand,
the $p$-values of the same test based on the outlier deleted data remain stable,
supporting rejection, at all values of $\beta$ (Figure \ref{fig:Tele_mu0_composite}(a)). The stable
behavior of the test statistic based on the density power divergence for the full data
approximately coincides with the stability of the density power divergence
estimate of $\mu$ itself, obtained under a two-parameter normal model, which
is presented in Figure \ref{fig:Tele_mu0_composite}(b). The minimum density power
divergence estimators of $\mu$ for the full data and the outlier deleted data are
practically identical for $\beta > 0.12$. At least in this example, the
robustness of the test statistic is clearly linked to the robustness of the
estimator.

\begin{figure}[tbp]
\centering
\includegraphics[
height=7.5278cm,
width=11.5339cm
]{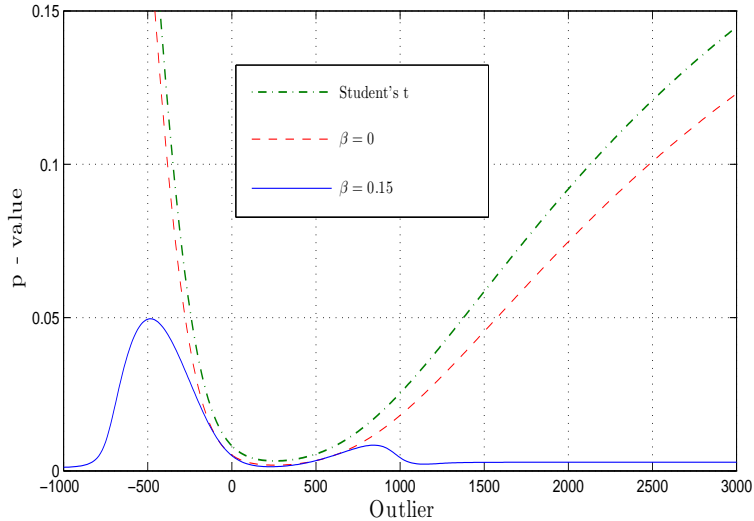}
\caption{Two sided $p$-values for the tests for the mean for the
telephone-fault data under the normal model against the first outlying
observation.}
\label{outliers}
\end{figure}

To further explore the robustness properties of the density power divergence
tests we look at the two sided $p$-values for different values of the
outlier. For this purpose we vary the first outlying observation in the
range from $-1000$ to 3000 by keeping the remaining 13 observations fixed.
Figure \ref{outliers} shows the corresponding $p$-values of the density
power divergence tests with $\beta = 0.15$ as well as $\beta=0$ and the
ordinary $t$-test. It shows that initially the $p$-value of the density
power divergence test with $\beta = 0.15$ increases as the first
observation moves away from the center of the data set, but after a certain
limit the test gradually nullifies the effect of the outlier. On the other
hand, the $p$-values of the $t$-test and the density power divergence test
with $\beta = 0$ keep on increasing with the outlier on either tail.
 Indeed the $p$-values of these two tests are remarkably
close to each other.


\subsubsection{Darwin's Plant Fertilization Data}

\label{SEC:Darwin_example}

Charles Darwin had performed an experiment which may be used to determine whether
self-fertilized plants and cross-fertilized plants have different growth
rates. In this experiment pairs of \textit{Zea mays} plants, one self and
the other cross-fertilized, were planted in pots, and after a specific time
period the height of each plant was measured. A particular sample of 15 such
pairs of plants led to the paired differences (cross-fertilized minus self
fertilized) presented in increasing order in Table \ref{TAB:Darwin-data}
(see  \citealp{darwin1878}).

\begin{table}[h]
\caption{Darwin's Plant Fertilization Data}
\label{TAB:Darwin-data}
\begin{center}
\begin{tabular}{lccccccccccccccc}
\hline
Pair & 1 & 2 & 3 & 4 & 5 & 6 & 7 & 8 & 9 & 10 & 11 & 12 & 13 & 14 & 15 \\
Difference & $-67$ & $-48$ & 6 & 8 & 14 & 16 & 23 & 24 & 28 & 29 & 41 & 49 &
56 & 60 & 75 \\ \hline
\end{tabular}%
\end{center}
\end{table}

As in the previous example, we assume a normal model for the paired
differences and test $H_0: \mu = 0$ against $H_1: \mu \neq 0$, i.e. we test whether
the mean of the paired differences is different from zero. The unconstrained
minimum DPD estimates of $\mu$ 
under the normal model corresponding to different values of the tuning
parameter $\beta$ are presented in Figure \ref{fig:Darwin}(b).
The two negative paired differences appear to be geometrically well
separated from the rest of the data, though they are perhaps not as huge
outliers as the first observation in the telephone-fault data. These two
observations do have a substantial impact on the parameter estimates and the
test statistic for testing $H_0$ using density power divergence tests with
very small values of $\gamma = \beta$, and it is instructive to compare to the
case where these two outliers have been removed from the data.
For small values of $\beta$, the two sided $p$-values of the test
statistics are drastically different for the full data and outlier deleted
cases (Figure \ref{fig:Darwin}(b)), but they get closer with increasing $%
\beta$, and they essentially coincide for $\beta \geq 0.45$. Once again
this seems to be directly linked to the robustness of the parameter
estimates; Figure \ref{fig:Darwin}(b), 
which also depicts the progression of the parameter estimates for the
outlier deleted data, clearly demonstrates that. For comparison we note that the two sided $p$-values for the ordinary $t$-test in this case are 0.0497 (for full data) and $1.3119 \times 10^{-4}$ (for the cleaned data with the two outliers removed).
\begin{figure}[tbp]
\centering%
\begin{tabular}{rl}
\includegraphics[height=7.5cm, width=8cm]{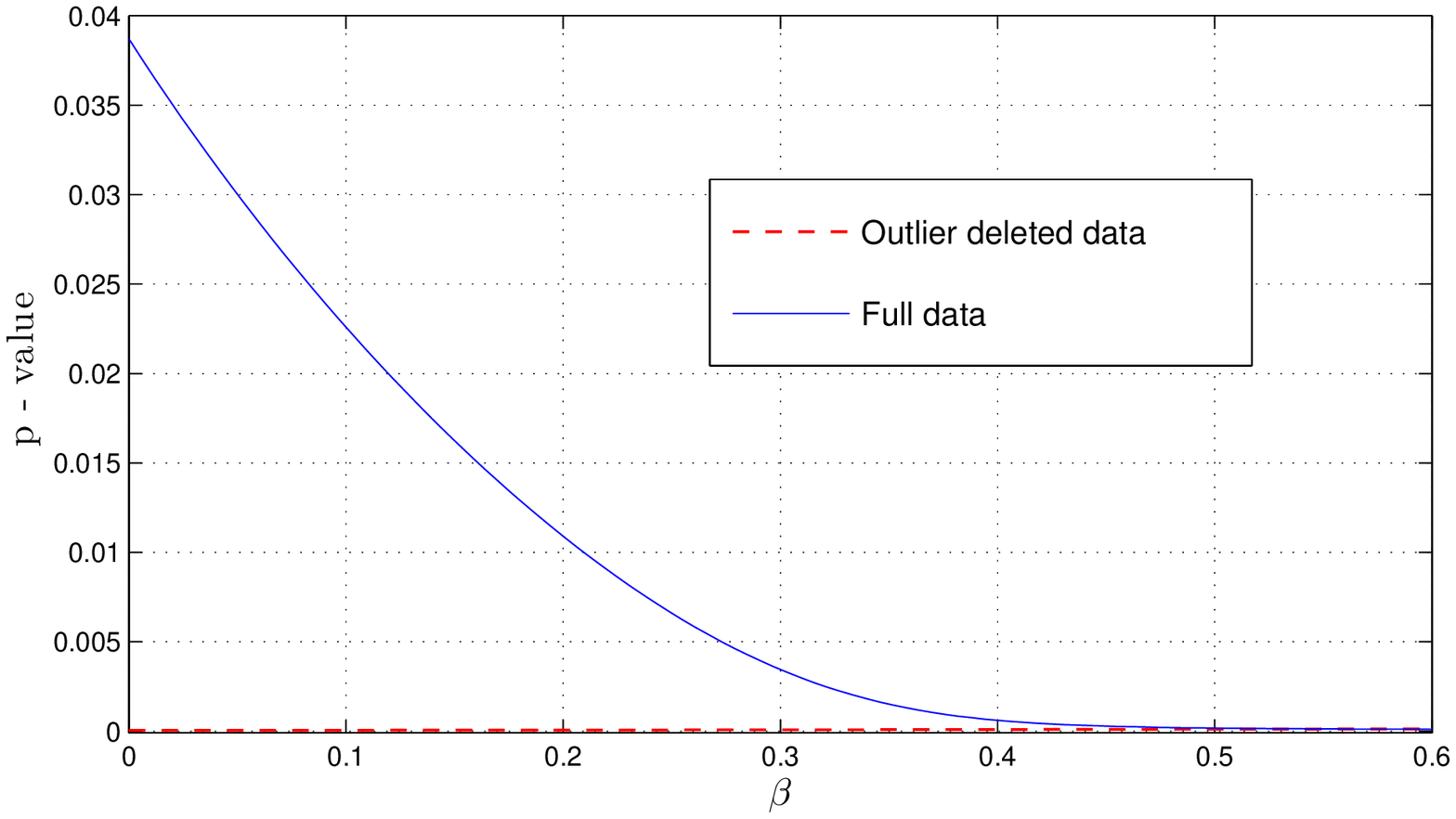}\negthinspace &
\negthinspace \includegraphics[height=7.5cm, width=8cm]{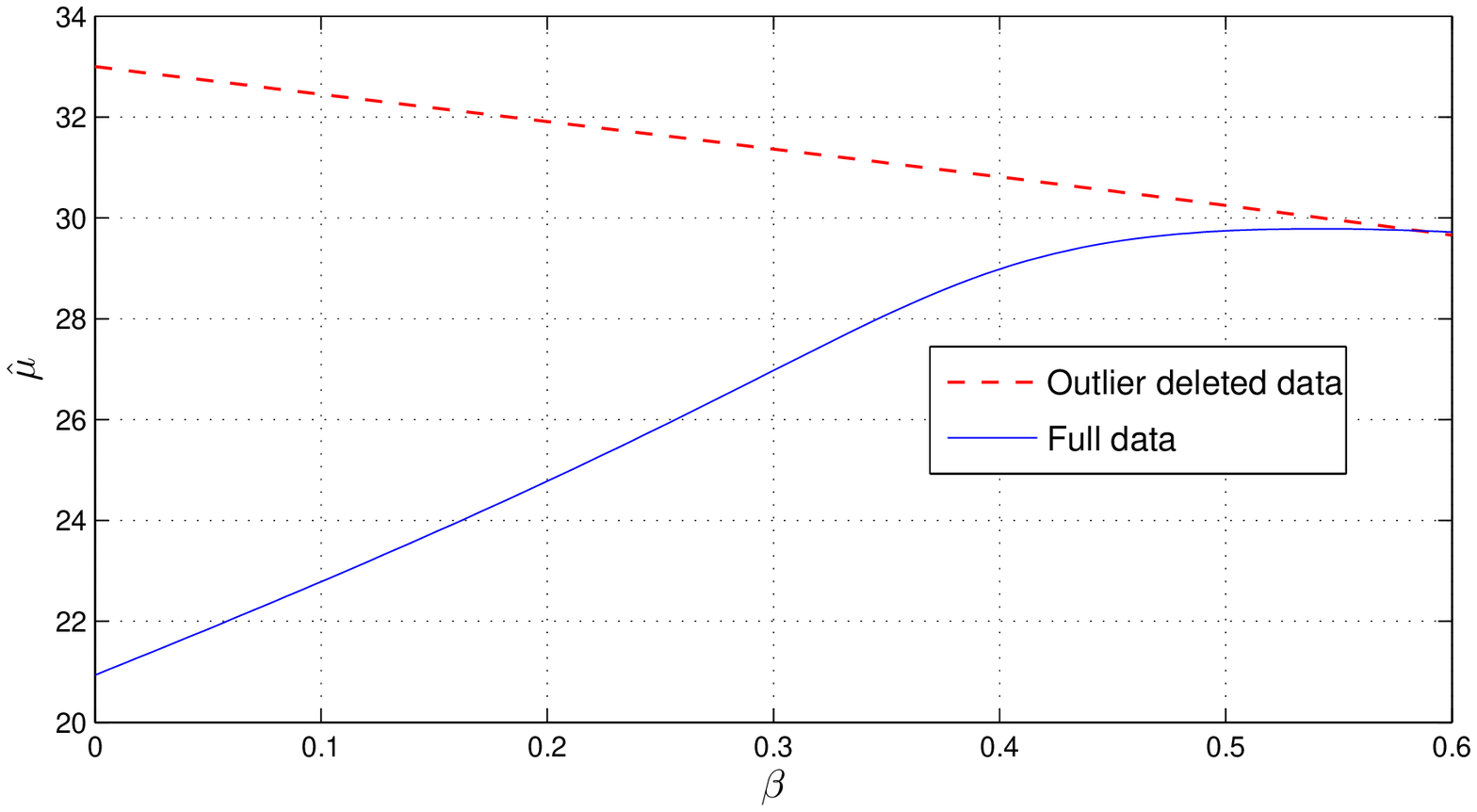} 
\end{tabular}%
\caption{(a) Two sided $p$-values of the density power divergence tests and 
(b) estimates of $\protect\mu$ for different values of $\protect\beta$ in case of Darwin's fertilization data.}
\label{fig:Darwin}
\end{figure}
%
%
%


\subsubsection{One Sided Tests}
\label{SEC:One_Sided}
In general, the default alternative hypotheses considered in our proposed tests are of the two sided type. Depending on the nature of the problem and the dimension of the parameter, one sided alternatives may sometimes be of interest. For the telephone fault data the primary interest could be in determining whether the mean fault rate is higher than zero (rather than simply whether it is different from zero). It is presumable that Darwin's interest in the fertilization problem was to determine whether cross fertilization leads to a higher growth rate compared to self fertilization; indeed the result of the test performed by R. A. Fisher (reported in \citealp{fisher1935design}) on the plant fertilization data relates to the one sided alternative. In this subsection we consider appropriate one sided tests for these two real data examples presented earlier in this section. For this purpose we consider the signed divergence statistic (the signed square root of the statistic presented in (\ref{dpd_stat_normal})) as was done in \cite{MR999667}. The relevant one sided $p$-values are determined using the normal approximation, or that based on the $t$-distribution. In the following we will describe the problem of testing $H_0: \mu = 0$ against $H_1: \mu > 0$ under the normal model with unknown scale. 

The formal theory of constrained statistical inference (see \citealp{MR2099529}) established the expression of the asymptotic likelihood ratio test for the hypotheses $H_0: \mu = 0$ against $H_1: \mu > 0$ to be
\begin{equation*}
T_{\gamma =0,\beta =0}^{(1)}=I(\bar{X}>0)T_{\gamma =0}(\boldsymbol{\hat{%
\theta}}_{\beta =0},\boldsymbol{\tilde{\theta}}_{\beta =0})=I(\bar{X}%
>0)n\log \left( \tfrac{\tilde{\sigma}_{n}^{2}}{\hat{\sigma}_{n}^{2}}\right) 
\end{equation*}%
with asymptotic distribution equal to $\tfrac{1}{2}\chi _{0}^{2}+\tfrac{1}{2}%
\chi _{1}^{2}$ under $H_0$, where $\chi _{0}^{2}=0$ a.s., $I(\cdot )$ is
the indicator function and $\boldsymbol{\hat{\theta}}_{\beta =0}$, $%
\boldsymbol{\tilde{\theta}}_{\beta =0}$ are respectively the MDPDE
and RMPDE for $\boldsymbol{\theta }=(\mu ,\sigma )^{T}$ when the parameter
spaces are the unrestricted and restricted ones of the two sided test (\ref{A}). This test is almost the same as the one provided by the signed divergence likelihood ratio test statistic, 
\begin{equation*}
\widetilde{T}_{\gamma =0,\beta =0}^{(1)}=sign(\bar{X})\sqrt{T_{\gamma =0}(%
\boldsymbol{\hat{\theta}}_{\beta =0},\boldsymbol{\tilde{\theta}}_{\beta =0})}%
=sign(\bar{X})\sqrt{n\log \left( \tfrac{\tilde{\sigma}_{n}^{2}}{\hat{\sigma}%
_{n}^{2}}\right) },
\end{equation*}%
since the corresponding $p$-values at $t=n\log \left( \frac{\tilde{\sigma}%
_{n}^{2}}{\hat{\sigma}_{n}^{2}}\right) $ are given by%
\begin{align*}
p\mbox{-value}_{T_{\gamma =0,\beta =0}^{(1)}}(I(\bar{x}>0)t)& =\frac{1}{2}\Pr (\chi _{1}^{2}>I(\bar{x}>0)t), \\
p\mbox{-value}_{\widetilde{T}_{\gamma =0,\beta =0}^{(1)}}\left( sign(\bar{x%
})\sqrt{t}\right) & =\Pr (Z>sign(\bar{x})\sqrt{t}),
\end{align*}%
where $Z$ follows standard normal distribution. This means that if $\bar{x}>0$, both $p$-values\ are equal, whereas for $\bar{x}\leq 0$, 
\[ p\mbox{-value}_{T_{\gamma =0,\beta
=0}^{(1)}}(I(\bar{x}>0)t)=1>p\mbox{-value}_{\widetilde{T}_{\gamma =0,\beta
=0}^{(1)}}\left( sign(\bar{x})\sqrt{t}\right) >\frac{1}{2}.
\]
Both tests are in practice equivalent, and such a difference for big $p$%
-values comes from the fact that $\widetilde{T}_{\gamma =0,\beta =0}^{(1)}$
is formally more appropriate for $H_{0}:\mu \leq 0$ against $H_{1}:\mu >0$.
We shall restrict ourselves, for simplicity, only to the signed divergence likelihood ratio test statistics and their DPD based analogues;
the latter class of signed divergence DPDTS may be defined as 
$$\widetilde{T}_{\gamma,\beta}^{(1)} = sign(\hat{\mu}_{\beta })\left\{ \frac{\widetilde{\sigma }_{\beta } ^{\gamma } \sqrt{\gamma +1}\left( 2\beta +1\right) ^{3/2} \left( 2\pi \right) ^{\gamma
/2}}{\left( \beta +1\right) ^{3}} T_{\gamma }(\boldsymbol{\hat{\theta}}_{\beta },\boldsymbol{\tilde{\theta}}_{\beta })\right\}^{-1/2}$$
with asymptotic
distribution equal to the standard normal under $H_0$. 
In calculating the one sided $p$-values based on the signed divergence Hellinger distance test in case of the telephone fault data, \cite{MR999667} used an approximation based on the $t$-distribution, as the sample size was only 14.  In large samples, the distribution of the statistic is approximately normal. 
The
problem for a normal distribution with dimension bigger than one with
inequality restrictions is more complicated and requires a specific theory
based on \cite{MR2099529}. \cite{martin2013hypothesis} 
illustrate the procedure of handing this problem when $\phi$-divergence based
test statistics and MLEs are applied.



\bigskip


\noindent {\bf Telephone-fault data} The one sided $p$-values for the signed divergence DPDTSs corresponding to $\beta = 0.15$ and 0.3 are presented in Table \ref{tab:tele} for the full as well as outlier deleted data, using both the standard normal 
($Z$) and $t$ (with suitable degrees of freedom) approximations. The result for the ordinary one-sided $t$-test are also presented for comparison. The presence of the large outlier masks the 
significance in case of the $t$-test, but the signed divergence DPDTSs provide consistent significant results with and without the outlier. Similar results were reported by \cite{MR999667} with the signed divergence Hellinger distance test. The mean of the ordered differences between the inverse test rates and inverse control rates does appear to be greater than zero.

\begin{table}
\caption{$p$-values of the one sided tests for the mean in case of the telephone-fault data.}
\label{tab:tele}
\begin{center}
\begin{tabular}{ccccccccc}
 \hline
        & \multicolumn{8}{c}{Scenario}\\ \cline{2-9}
cutoff & \multicolumn{2}{c}{$t$-test} & & \multicolumn{2}{c}{DPD(0.15)} & & \multicolumn{2}{c}{DPD(0.3)}\\ \cline{2-3} \cline{5-6} \cline{8-9} 
       & Full & Deleted & & Full & Deleted & & Full & Deleted\\
\hline
$Z$ & -- & -- & & 0.0006 & 0.0019 & & 0.0013 & 0.0017\\
$t$ & 0.3481 & 0.0037 & & 0.0032 & 0.0068 & & 0.0050 & 0.0064\\
\hline
\end{tabular} 
\end{center}
\end{table}

\bigskip

\noindent {\bf Darwin's plant fertilization data} The results are presented in Table \ref{tab:dar}. The full data $p$-value was reported by \cite{fisher1935design}. In this case the one sided $p$-values for the $t$-test  lead to a shift from marginal significance to solid rejection due to the deletion of the (two) outliers. This also seems 
to be the case for signed divergence DPDTSs for very small values of $\beta$. However, larger values of $\beta$ lead to a more consistent behavior of the tests. This dataset requires stronger downweighting compared to the telephone-fault data, as the outliers here are less extreme, and therefore more difficult to identify. Under a suitable robust test, it appears that the mean growth of cross fertilized plants would be declared to be significantly higher than self fertilized plants.

\begin{table}
\caption{$p$-values of the one sided tests for the mean in case of the Darwin's plant fertilization data.}
\begin{center}
\begin{tabular}{ccccccccc}
 \hline
        & \multicolumn{8}{c}{Scenario}\\ \cline{2-9}
cutoff & \multicolumn{2}{c}{$t$-test} & & \multicolumn{2}{c}{DPD(0.15)} & & \multicolumn{2}{c}{DPD(0.3)}\\ \cline{2-3} \cline{5-6} \cline{8-9} 
       & Full & Deleted & & Full & Deleted & & Full & Deleted\\
\hline
$Z$ & -- & -- & & 0.0081 & $< 10^{-4}$ & & 0.0017 & $< 10^{-4}$\\
$t$ & 0.0252 & $< 10^{-4}$ & & 0.0153 & 0.0008 & & 0.0055 & 0.0009\\
\hline
\label{tab:dar}
\end{tabular} 
\end{center}
\end{table}


\subsection{Simulation Results}
\label{SEC:simulation}

\subsubsection{Normal Case} \label{sec:sim_mormal}

To further explore the performance of our proposed test statistic in case of
the $\mathcal{N}(\mu ,\sigma ^{2})$ problem, we studied the behavior of the
tests through simulation. We considered the hypothesis $H_{0}:\mu =0$
against the alternative $H_{1}:\mu \neq 0$ with $\sigma ^{2}$ unknown when
data were generated from the $\mathcal{N}(0,1)$ distribution. Subsequently,
the same hypotheses were tested when the data were generated from the $%
\mathcal{N}(1,1)$ distribution. In the first case our interest was in
studying the observed level (measured as the proportion of test statistics
exceeding the chi-square critical value in a large number -- here 10000 --
of replications) of the test under the correct null hypothesis, and in the
second case we were interested in the observed power (obtained in a similar manner as
above) of the test under the incorrect null hypothesis. The results are
given in Figures \ref{fig:normal}(a) and \ref{fig:normal}(b). In either case the
nominal level was $0.05$. We have used the ordinary $t$-test together with several DPD test statistics, corresponding to $\beta = 0, 0.1, 0.15$ and $0.25$, in this particular study. The horizontal lines in Figure \ref{fig:normal}(a), and later in Figure \ref{fig:normal}(c), represent the nominal level of 0.05. 

It may be noticed that all the tests excepting the exact likelihood ratio test
(the $t$-test) are slightly liberal for very small sample sizes and lead to
somewhat inflated observed levels. However this discrepancy decreases
rapidly, and by the time the sample size is 30 or more the  observed
levels have settled down reasonably around acceptable values.  The observed powers of the tests as given in
Figure \ref{fig:normal}(b) are, in fact, extremely close; in very small sample
sizes the other tests have slightly higher power than the $t$-test, but this
must be a consequence of the observed levels of these tests being higher
than the latter for such sample sizes. On the whole the proposed tests appear to
be quite competitive to the ordinary $t$-test for pure normal data.

To evaluate the stability of the level of the tests under contamination, we repeated the tests for $H_0: \mu = 0$ against $H_1: \mu \neq 0$ under data generated from the mixture of $N(0, 1)$ and $N(-10, 1)$, where the mixing weight of the first component is 0.9. To illustrate the stability of power, the tests were performed with data generated under a mixture of $N(1, 1)$ and $N(-10, 1)$, where the mixing weight of the first component is again 0.9. The results are given in \ref{fig:normal}(c) and \ref{fig:normal}(d) respectively.


In this case there is a drastic and severe inflation in the observed level
of the $t$-test and that of the DPD(0) test. As $\beta$ increases, however,
the resistant nature of the tests are clearly apparent. By the time $\beta
= 0.25$, the levels have already been reduced to acceptable values. The
opposite behavior is seen in case of power. There appears to be a complete
breakdown in power for small values of $\beta$, but the power remains quite
stable for values of $\beta$ equal to 0.25 or greater.

On the whole it appears to be fair to claim that for sample sizes equal to
or larger than 30 the efficiency of many of our DPDTSs are very close to
the efficiency of the $t$-test, but the robustness properties of our tests
are often significantly better than the $t$-test in terms of maintaining the
stability of both the level and power.

\begin{figure}[tbp]
\centering%
\begin{tabular}{rl}
\includegraphics[height=7.5cm, width=8cm]{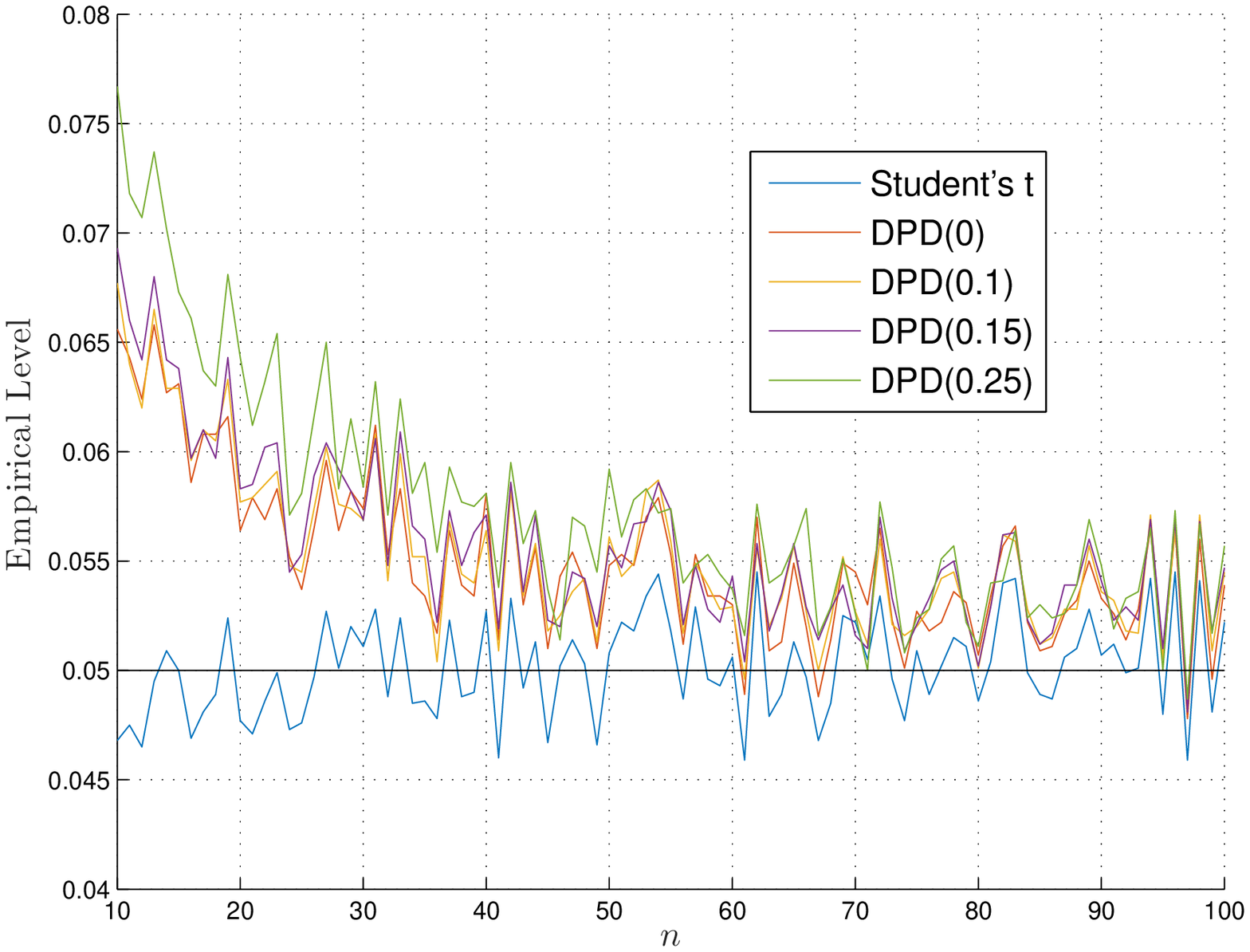}\negthinspace &
\negthinspace \includegraphics[height=7.5cm, width=8cm]{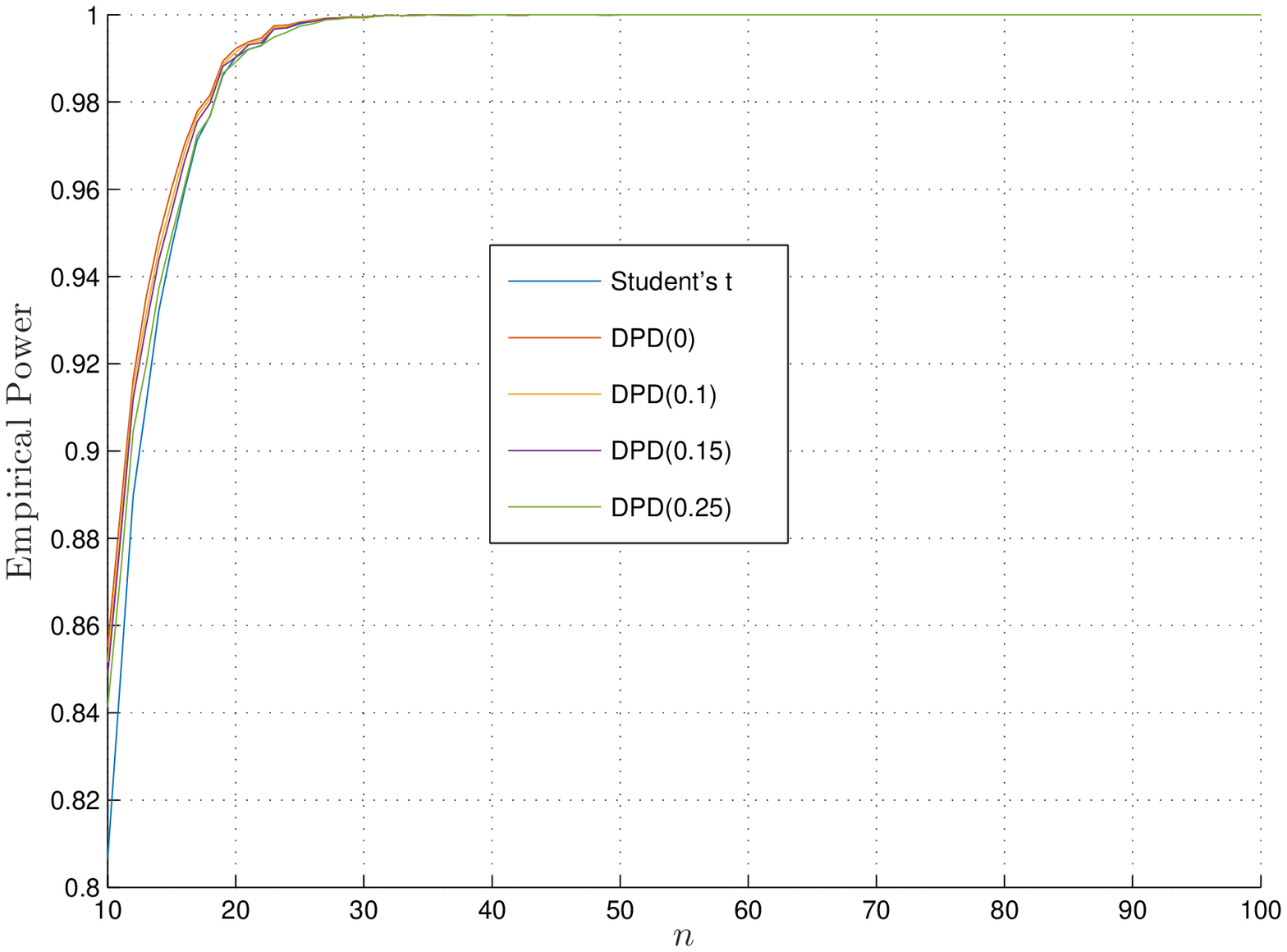} \\
\multicolumn{1}{c}{\textbf{(a)}} & \multicolumn{1}{c}{\textbf{(b)}} \\
\includegraphics[height=7.5cm, width=8cm]{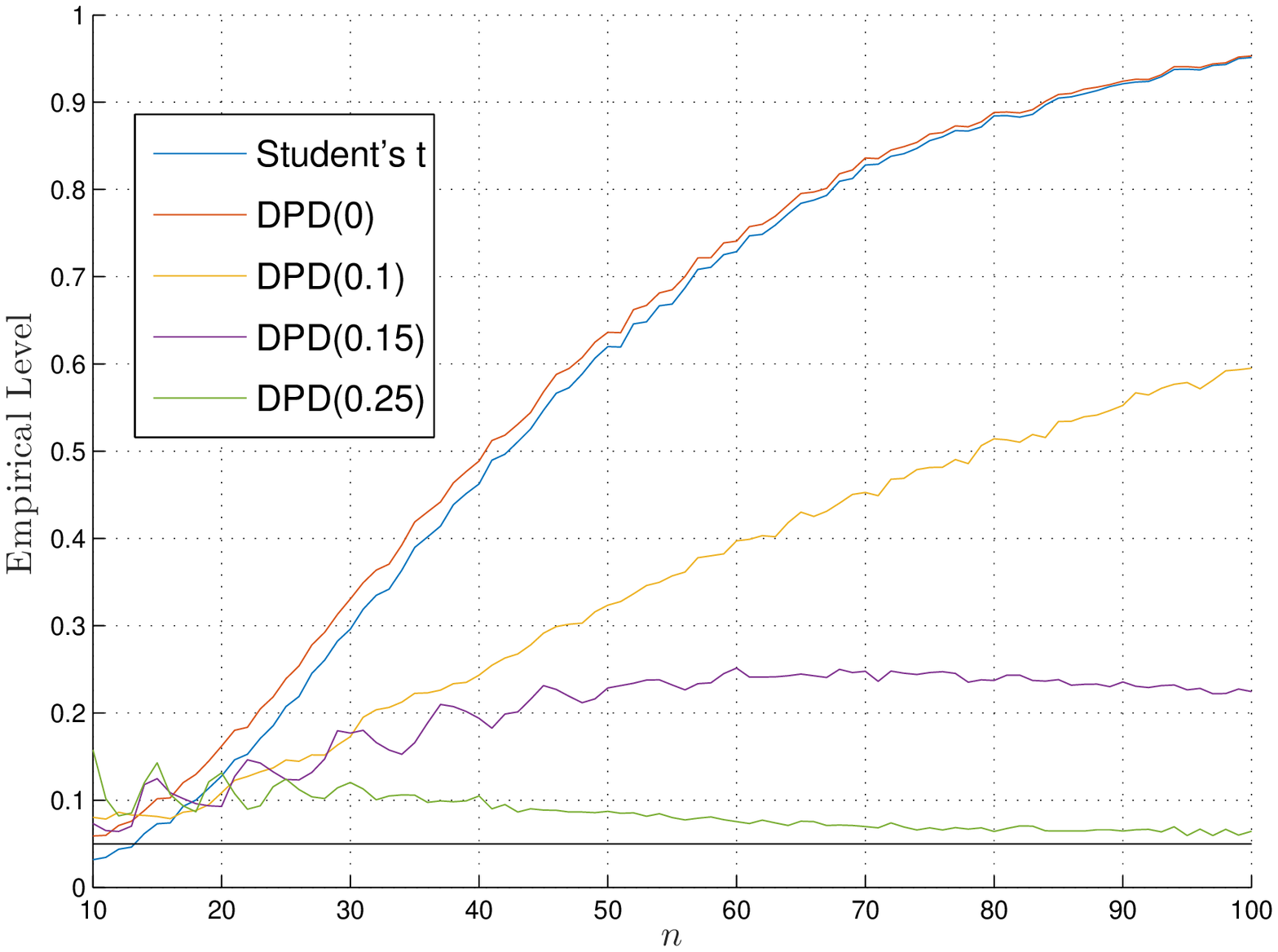}\negthinspace &
\negthinspace \includegraphics[height=7.5cm, width=8cm]{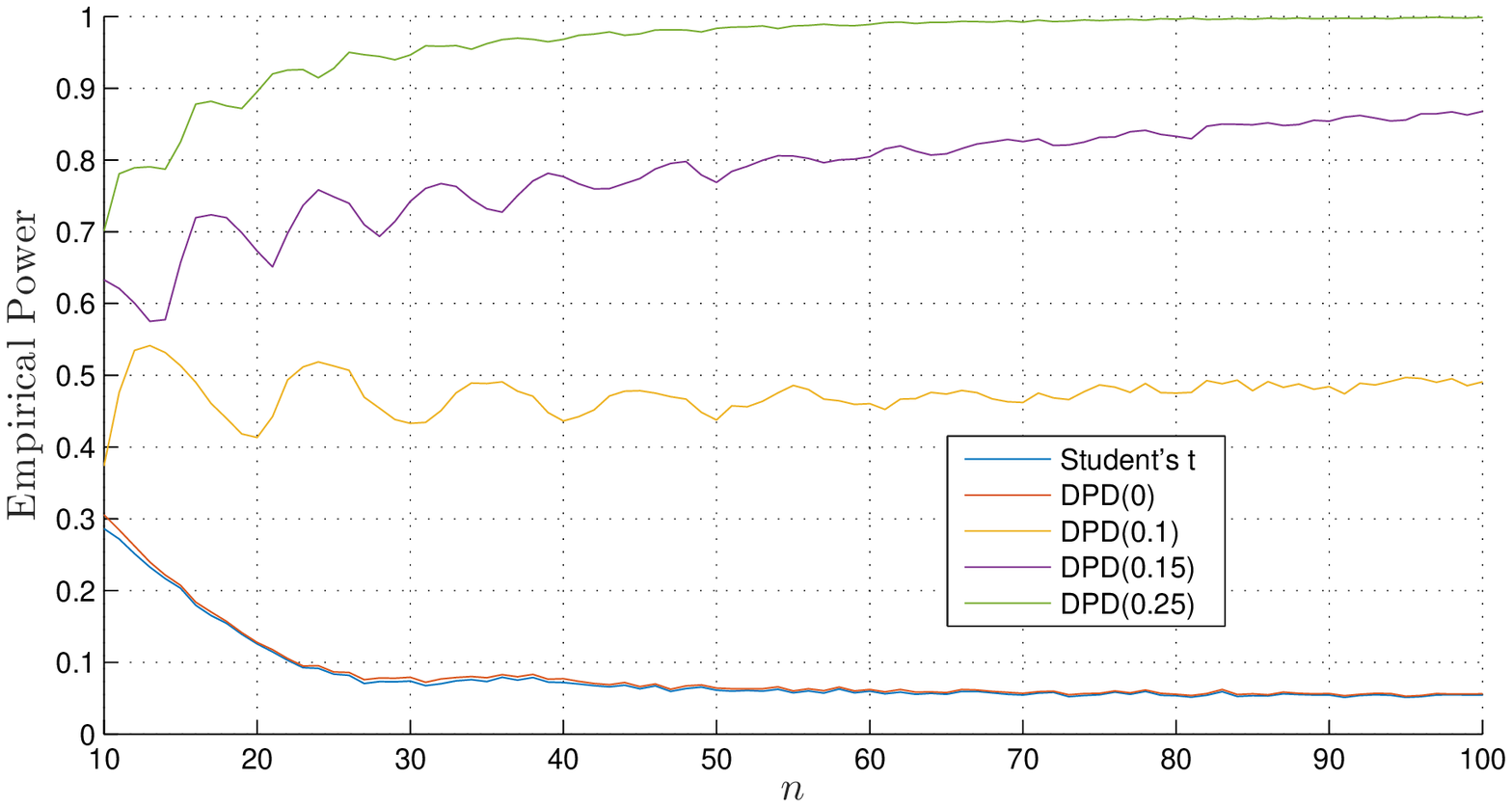} \\
\multicolumn{1}{c}{\textbf{(c)}} & \multicolumn{1}{c}{\textbf{(d)}}%
\end{tabular}%
\caption{Simulated levels and powers of the DPDTSs for pure and contaminated
data in case of the normal distribution.}
\label{fig:normal}
\end{figure}

\begin{figure}
\centering%
\begin{tabular}{rl}
\includegraphics[height=7.5cm, width=8cm]{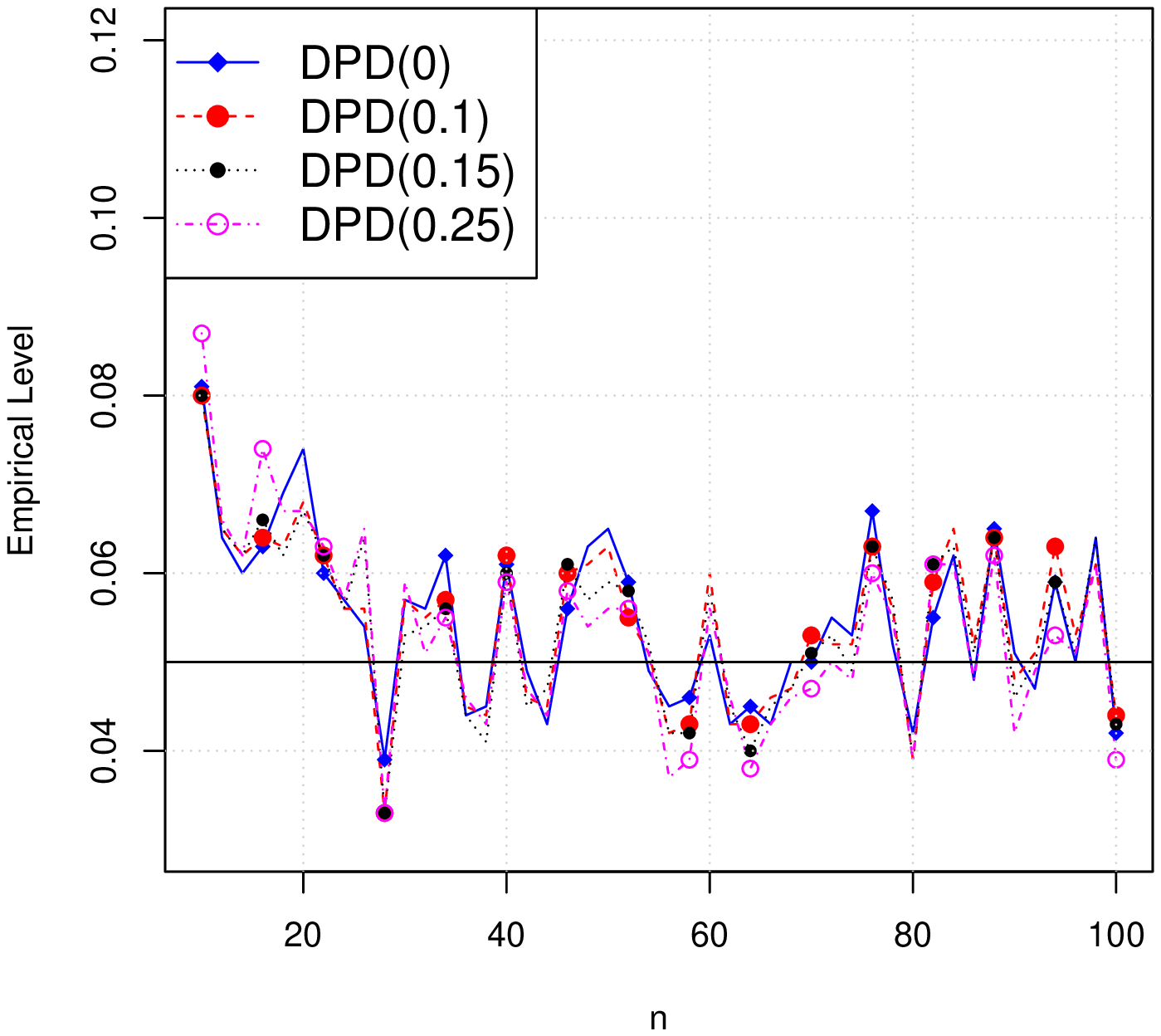}\negthinspace &
\negthinspace \includegraphics[height=7.5cm, width=8cm]{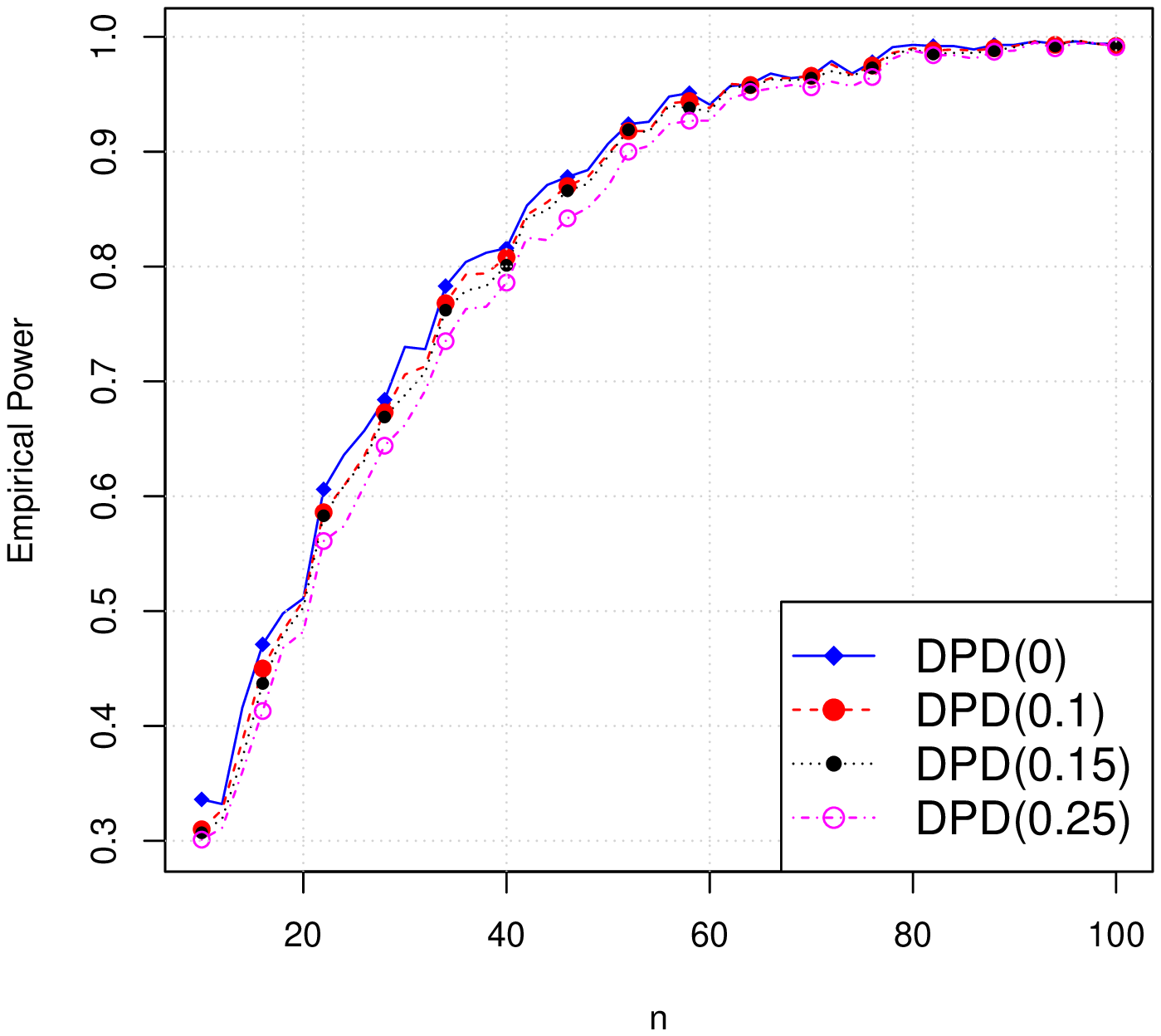} \\
\multicolumn{1}{c}{\textbf{(a)}} & \multicolumn{1}{c}{\textbf{(b)}} \\
\includegraphics[height=7.5cm, width=8cm]{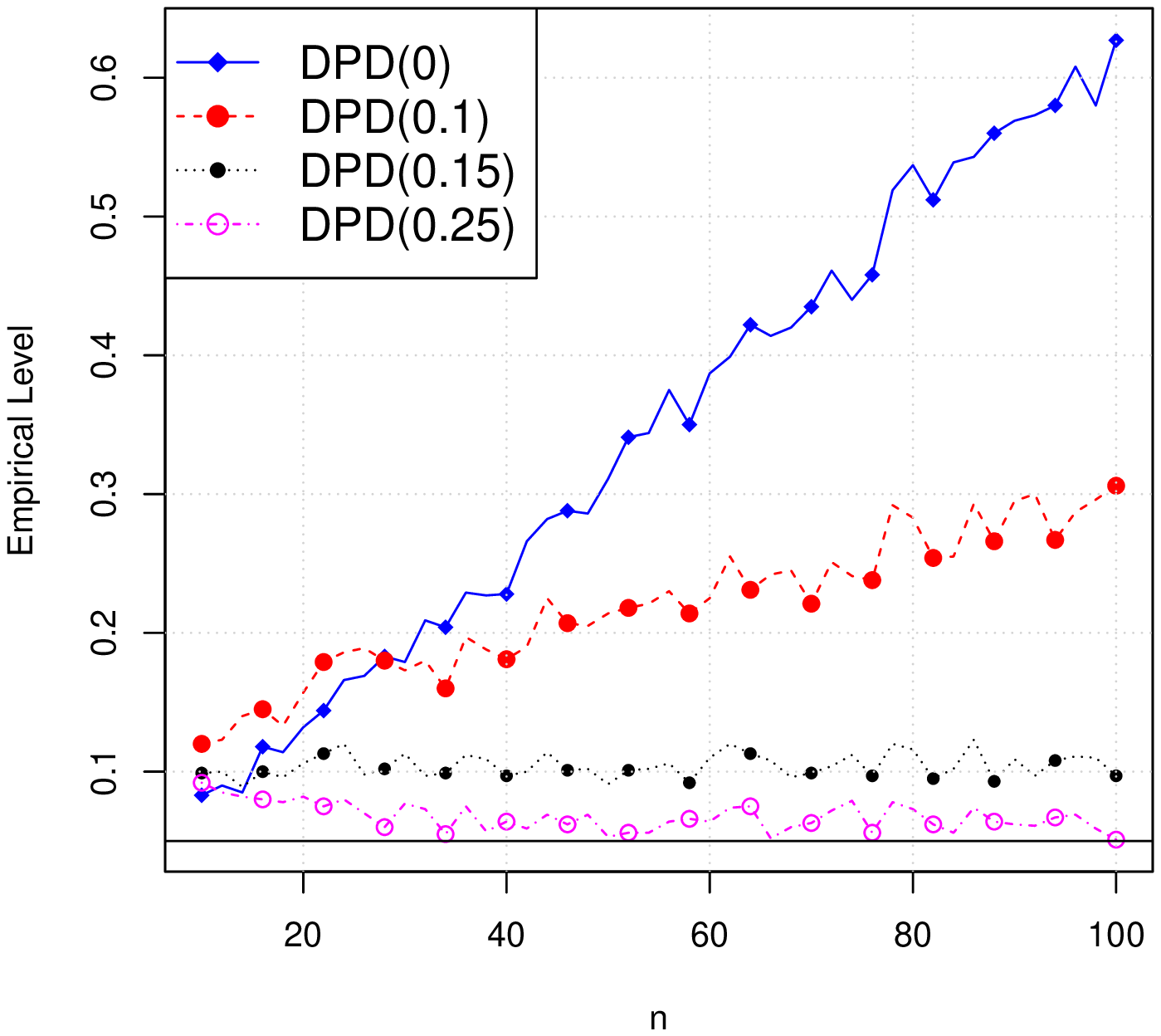}\negthinspace &
\negthinspace \includegraphics[height=7.5cm, width=8cm]{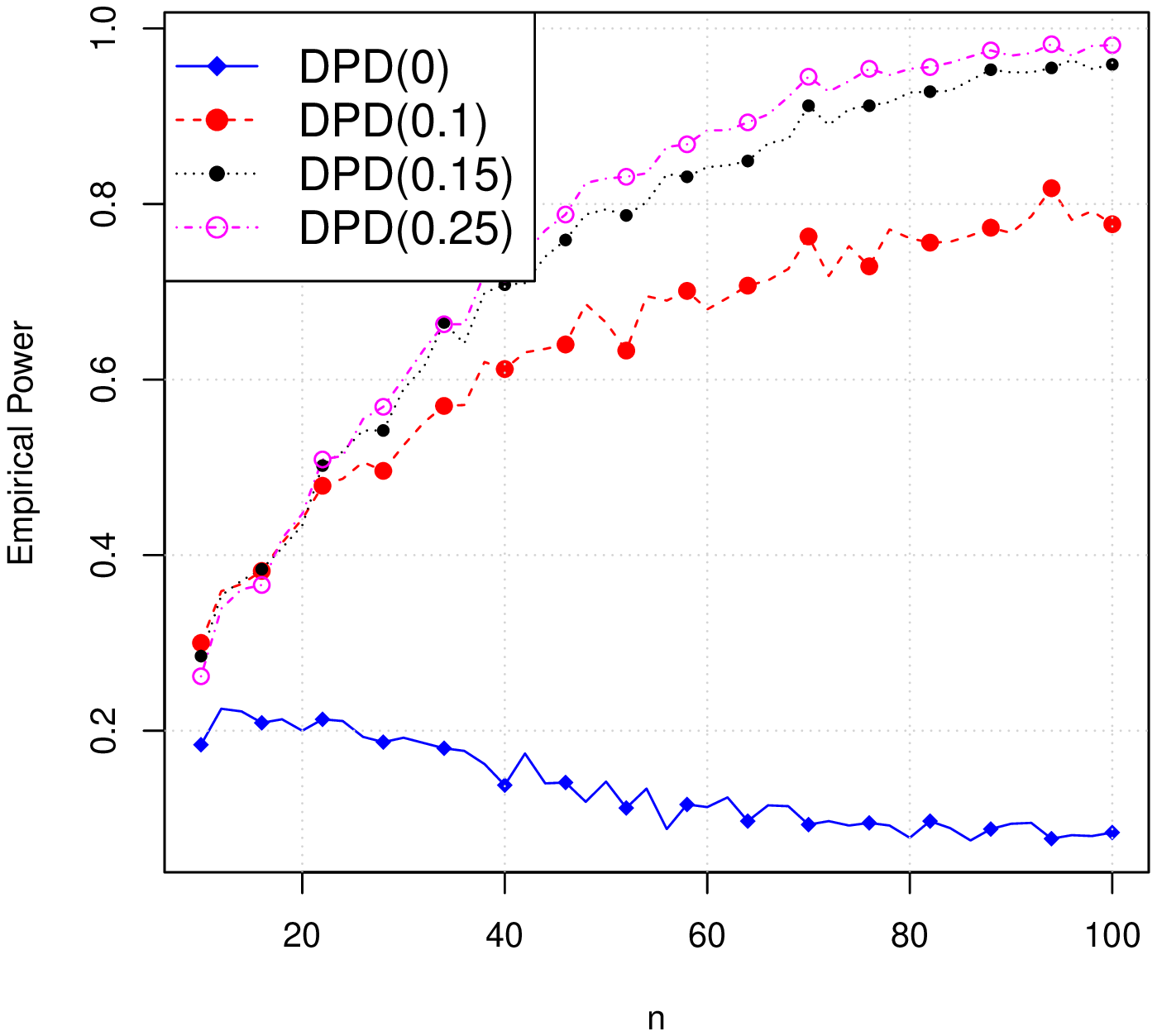} \\
\multicolumn{1}{c}{\textbf{(c)}} & \multicolumn{1}{c}{\textbf{(d)}}%
\end{tabular}%
\caption{Simulated levels and powers of the DPDTSs for pure and contaminated
data in case of the Weibull distribution.}
\label{fig:weibull}
\end{figure}
\subsubsection{Weibull Case}
As we have mentioned before, it is important to demonstrate the properties of the proposed method in models other than the normal so that one has a better idea about the scope of the method. Accordingly we performed tests of composite hypotheses under the Weibull model in the spirit of Section \ref{sec:sim_mormal}. 
Let us consider the hypothesis defined in (\ref{w}), where $\sigma_0$ is taken to be 1.5. In the first study we have 
generated data from the $\mathcal{W}(1.5,1.5)$ distribution. The plot for the observed level for the hypothesis $H_0:\sigma = 1.5$ against the two sided alternative is given in Figure \ref{fig:weibull}(a), where we have used 1,000 replications. Next
the same hypotheses were tested when the data were generated from the $%
\mathcal{W}(1.1,1.5)$ distribution. The observed power function is plotted in Figure \ref{fig:weibull}(b) for different values of $\beta$. The powers are remarkably close. In all cases the nominal level was $0.05$.

To evaluate the stability of the level and the power of the tests under
contamination, we repeated the tests with data generated from the Weibull mixture consisting of 95\% $\mathcal{W}(1.5,1.5)$ and 5\% $\mathcal{W}%
(25,1.5)$, and then from mixture of 95\% $\mathcal{W}(1.1,1.5)$ and 5\% $\mathcal{W}(25,1.5)$. In either case the first larger component is our target.
 In Figures \ref{fig:weibull}(c), the levels of the statistics under the contamination of first type are presented indicating the stability of levels for moderately large values of $\beta$. Figure \ref{fig:weibull}(d) demonstrates the stability of powers under contaminated data of the second type for the same values of $\beta$.

 \subsubsection{Comparison with Other Robust Tests}

Here we provide a comparison of our proposed tests with some other popular resistant tests available in the literature. In particular we have used a parametric test -- the Winsorized test of \cite{dixon1968approximate} together with three nonparametric tests -- the one sample Kolmogorov-Smirnov (KS) test, the two sided Wilcoxon signed rank test and the two sided sign test. 
The model, the hypotheses, the parameters chosen, the level of significance and other details of the set up of this simulation are the same as those in Section \ref{sec:sim_mormal}. 

We have Winsorized the 15\% extreme observations on each tail of the data distribution in case of the Winsorized $t$-test. 
Note that the null hypotheses are slightly different for the nonparametric tests. For the KS-test we first standardize the data using robust statistics, and then test whether the corresponding distribution is a standard normal. The data are standardized using the transformation $Z = (X - \mu_0)/\mbox{MAD}.$ Here $\mu_0$ is the null value and ${\rm MAD}$ is $1.4826 \times$(median absolute deviation about the median). In case of the Wilcoxon test and the sign test we perform tests for the population median without making any parametric model assumptions. For comparison just one DPDTS is used in these simulations, that corresponding to the tuning parameter $\beta = 0.25$. To emphasize the robustness properties of these tests we have also included the Student's $t$-test in these investigations, so that the robust tests stand out in contrast. Our simulation results are presented in Figure \ref{fig:normal_comp}. 

From Figure \ref{fig:normal_comp}(a) it may be observed that the empirical levels of the Winsorized $t$-test, the KS-test  and the Wilcoxon test are very close to the nominal level for pure normal data. For small sample sizes the DPDTS is slightly liberal; however even at a sample size of 30, it is off by only one percent compared to the nominal level. On the other hand the sign test is a bit conservative, even at fairly large samples. The observed powers of all the tests in Figure \ref{fig:normal_comp}(b) rapidly approach unity in fairly small samples. The results in Figure \ref{fig:normal_comp}(c) demonstrate that for contaminated data all tests except the DPDTS fail to maintain the nominal level. The observed level of the sign test is close to the nominal level for small sample sizes but eventually as the sample size increases it also breaks down. The powers of the tests for the contaminated data, plotted in Figure \ref{fig:normal_comp}(d) show that  all the robust tests exhibit stable power. For small sample sizes the DPDTS exhibits the highest power. The overall observation on the basis of all the above appears to be that the DPD based test is superior to the classical Wald test under contamination, and is also competitive or better than several other standard resistant tests in terms of robustness, at least to the extent this particular simulation study is concerned. 


\begin{figure}
\centering%
\begin{tabular}{rl}
\includegraphics[height=7.5cm, width=8cm]{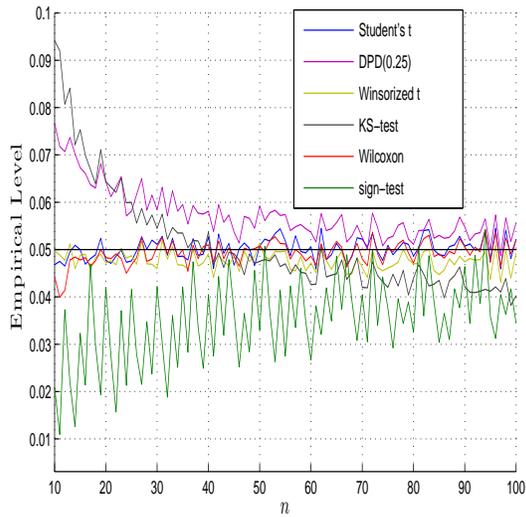}\negthinspace &
\negthinspace \includegraphics[height=7.5cm, width=8cm]{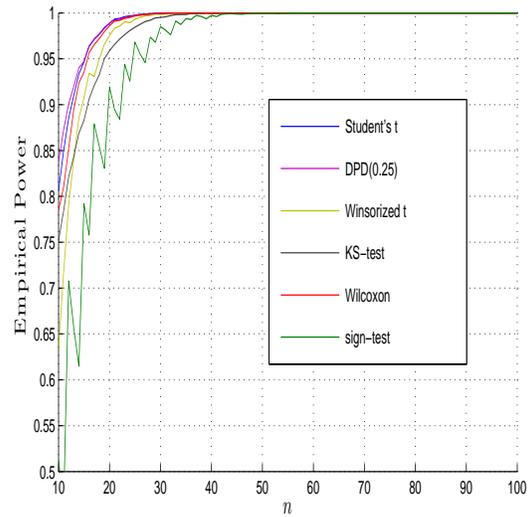} \\
\multicolumn{1}{c}{\textbf{(a)}} & \multicolumn{1}{c}{\textbf{(b)}} \\
\includegraphics[height=7.5cm, width=8cm]{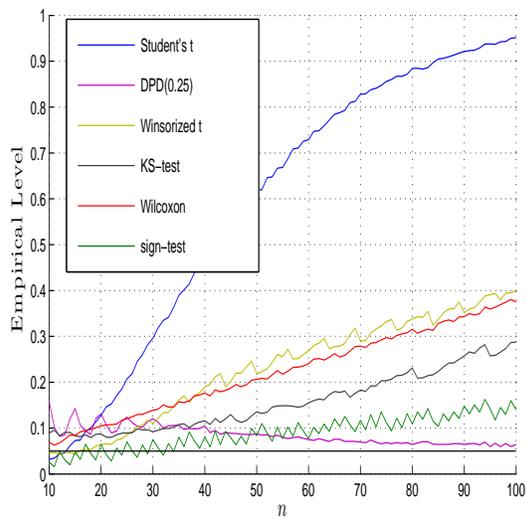}\negthinspace &
\negthinspace \includegraphics[height=7.5cm, width=8cm]{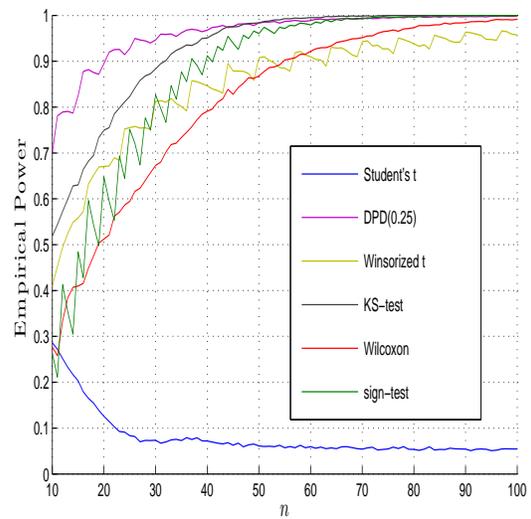} \\
\multicolumn{1}{c}{\textbf{(c)}} & \multicolumn{1}{c}{\textbf{(d)}}%
\end{tabular}%
\caption{Simulated levels and powers of some robust tests for pure and contaminated
data in case of the normal distribution.}
\label{fig:normal_comp}
\end{figure}

\section{Choosing the Tuning Parameter}\label{tuning}

By construction, the test statistic in (\ref{TRES.1}) employs two different tuning parameters $\beta$ and $\gamma$. These
parameters have two different roles, in two different stages in the hypothesis testing process. The parameter 
$\beta$ is used to evaluate the robust unconstrained and constrained (under the null hypothesis) estimators. In the 
next stage a density power 
divergence with parameter $\gamma$ is constructed to quantify the disparity between the fitted unrestricted and the 
restricted models. As seen in Theorem \ref{Theorem3}, the null distribution of the statistic can be derived for all values of 
the parameters $\beta, \gamma > 0$, and in practice one can choose them independently of one another. As the 
robustness of the test statistic depends primarily on the robustness of the estimators, the choice of the parameter 
$\beta$ turns out to be more critical in our testing procedure. In repeated simulations
(not presented here) our observation is that the parameter $\gamma$ does not have a significant impact on the 
robustness of the procedure. Thus while the generality of the method allows us the choice of possibly different 
tuning parameters, for simplicity of implementation we will let $\beta = \gamma$, so that the selection problem 
reduces to that of a single parameter. Throughout the paper, we have used $\beta = \gamma$ in our simulations and 
real data examples. 

In a real situation, the experimenter will require some guidance on the choice of this single tuning parameter $\beta$. \cite{broniatowski2012decomposable} have reported that values of $\beta \in [0.1, 0.25]$ are often reasonable choices; we largely agree with this view, although tentative outliers and heavier contamination may require greater downweighting through a larger value of $\beta$; this is the case, for example, in Darwin's plant fertilization data example. However, apart from fixed choices, other data driven and adaptive choices could also be useful, as one can then tune the parameter to make the procedure more robust as required.
%
In this paper we follow the approach of \cite{Warwick} for this purpose, which minimizes an empirical 
measure of the mean square error of the estimator to determine the ``optimal'' tuning parameter. This requires the use
of a robust pilot estimator of the parameter. \cite{Warwick} suggested the use of the MPDPDE corresponding
to $\beta = 1$. The optimal parameter depends on the choice of the pilot estimator, however, and as larger values of
$\beta$ lead to a loss in efficiency, \cite{MR3117102} suggested the choice of the MDPDE with $\beta = 0.5$ as the 
pilot estimator, which appears to be reasonable in most cases. Our subsequent analysis is based on the \cite{Warwick} method, with the \cite{MR3117102} modification. 

We do acknowledge that the criterion to be considered for the choice of the optimal $\beta$ for the testing problem is 
not necessarily the same for the estimation problem. In hypothesis testing the appropriate criterion should involve a
suitable linear combination of the inflation in the observed level under the null and the drop in power under contiguous
alternatives in a contaminated scenario. However, an appropriate measure of this sort is not easy to construct. As it 
appears that the robustness of the proposed tests correspond almost exactly to the robustness of the MDPDEs, we feel 
that the optimal choice of $\beta$ as described in the previous paragraph would generally work reasonably well in case 
of the hypothesis testing problem also. As of now, we recommend the choice of $\beta$ according to the above recipe. 

The above criterion leads to estimated optimal choices of $\beta$ to
be 0.1919 for the telephone fault data, and 0.5657 for Darwin's plant fertilization data respectively. As the first 
observation in the telephone fault data is a massive outlier, it is easily recognized by the testing procedures
even at fairly small values of $\beta$. However for Darwins' plant fertilization data the outliers are more tentative, 
and therefore require stronger downweighting to eliminate their effect.  

\section{Concluding Remarks}

\label{SEC:concluding}


This paper provides the appropriate theoretical machinery to perform general
parametric tests of hypotheses based on the density power divergence.
We demonstrate that one can construct a class of parametric tests of 
hypotheses based on the above measure which allows the experimenter to test for composite null hypotheses
 under the presence of nuisance parameters. The tests of this class have been shown to have excellent robustness properties in simulation studies and have a huge scope of application; for the purpose of numerical demonstration we have
chosen the scenario of the usual $t$-test and illustrated that for this
situation the proposed test provides extremely satisfactory results.
%
Similar
improvements are also demonstrated outside the normal model, when data are 
generated from the Weibull distribution. When
considered with the benefit of not requiring any intermediate smoothing
technique as in the case of the Hellinger deviance test, our proposed techniques appear
to prominently stand out among classes of robust tests for composite
hypotheses. Our results also appropriately generalize the results of  \cite{MR3011625}.


\bigskip

\noindent\textbf{Acknowledgments }This work was partially supported by Grant MTM-2012-33740.

\bibliographystyle{abbrvnat}
\bibliography{reference}

\section*{Appendix}
There is some overlap between the Lehmann and Basu et al. conditions. In the following we present the consolidated set of
conditions which are the useful ones in our context. 

\bigskip

\noindent \textbf{Lehmann and Basu et al. conditions}
\begin{itemize}
 \item[(LB1)] The model distributions $F_{\boldsymbol{\theta }}$ of $X$ have common support, so that the set $\mathcal{X} =
\{x|f_{\boldsymbol{\theta }} (x) > 0\}$ is independent of $\boldsymbol{\theta }$. The true distribution $H$ is also supported
on $\mathcal{X}$, on which the corresponding density $h$ is greater than zero.

 \item[(LB2)] There is an open subset of $\omega$ of the parameter space $\Theta$, containing the best fitting parameter $\boldsymbol{\theta }_0$ such that for almost all $x \in \mathcal{X}$, and all $\boldsymbol{\theta } \in \omega$,
the density $f_{\boldsymbol{\theta }} (x)$ is three times differentiable with respect to $\boldsymbol{\theta }$ and the
third partial derivatives are continuous with respect to $\boldsymbol{\theta }$.
 \item[(LB3)] The integrals $\int f_{\boldsymbol{\theta }}^{ 1+\beta} (x)dx$ and $\int f_{ \boldsymbol{\theta }}^\beta (x)h(x)dx$ can be differentiated three
times with respect to $\boldsymbol{\theta }$, and the derivatives can be taken under the
integral sign.

 \item[(LB4)] 
 The $p \times p$ matrix $\boldsymbol{J}_\beta(\boldsymbol{\theta })$, defined  in (\ref{2.3}),
is positive definite. 

 \item[(LB5)] There exists a function $M_{jkl} (x)$ such that
$|\nabla_{jkl} V_{ \boldsymbol{\theta }} (x)| \leq M_{jkl} (x)$ for all $\boldsymbol{\theta } \in \omega$,
where $E_h [M_{jkl} (X)] = m_{jkl} < \infty$ for all $j$, $k$ and $l$, where $V_\theta(x)$ is as defined in (\ref{B2}).
\end{itemize}

\noindent \textbf{Proof of Theorem \ref{theorem:dpd} } 
This proof closely follows the approach of \cite{MR2566692}. Let%
\begin{equation}
h_{n}(\boldsymbol{\theta })=\frac{1}{1+\beta }\left[ \int f_{\boldsymbol{%
\theta }}^{1+\beta }(x)dx-\left( 1+\frac{1}{\beta }\right) \frac{1}{n}%
\sum_{i=1}^{n}f_{\boldsymbol{\theta }}^{\beta }(X_{i})\right]  \label{EQ:Hn}
\end{equation}%
%
%
%
be the function (\ref{B2}) divided by $1+\beta $. By differentiating both
sides of equation (\ref{EQ:Hn}) with respect to $\boldsymbol{\theta }$ we
get
\begin{equation}
\frac{\partial }{\partial \boldsymbol{\theta }}h_{n}(\boldsymbol{\theta }%
)=\int \boldsymbol{u}_{\boldsymbol{\theta }}(x)f_{\boldsymbol{\theta }%
}^{1+\beta }(x)dx-\frac{1}{n}\sum_{i=1}^{n}\boldsymbol{u}_{\boldsymbol{%
\theta }}(X_{i})f_{\boldsymbol{\theta }}^{\beta }(X_{i})  \notag
\end{equation}%
%
%
and differentiating again with respect to $\boldsymbol{\theta }$%
\begin{align}
\frac{\partial }{\partial \boldsymbol{\theta }^{T}}\frac{\partial }{\partial
\boldsymbol{\theta }}h_{n}(\boldsymbol{\theta })& =(1+\beta )\int
\boldsymbol{u}_{\boldsymbol{\theta }}(x)\boldsymbol{u}_{\boldsymbol{\theta }%
}^{T}(x)f_{\boldsymbol{\theta }}^{1+\beta }(x)dx-\int \boldsymbol{I}_{%
\boldsymbol{\theta }}(x)f_{\boldsymbol{\theta }}^{1+\beta }(x)dx
\notag \\
& -\frac{\beta }{n}\sum_{i=1}^{n}\boldsymbol{u}_{\boldsymbol{\theta }}(X_{i})%
\boldsymbol{u}_{\boldsymbol{\theta }}^{T}(X_{i})f_{\boldsymbol{\theta }%
}^{\beta }(X_{i})+\frac{1}{n}\sum_{i=1}^{n}\boldsymbol{I}_{\boldsymbol{%
\theta }}(X_{i})f_{\boldsymbol{\theta }}^{\beta }(X_{i}). \notag
\end{align}%
Here $\boldsymbol{u}_{{\boldsymbol{\theta }}}(x)=\frac{\partial }{\partial
\boldsymbol{\theta }}\log f_{{\boldsymbol{\theta }}}(x)$ and $%
\boldsymbol{I}_{\boldsymbol{\theta }}(x)=-\frac{\partial }{\partial
\boldsymbol{\theta }}\boldsymbol{u}_{\boldsymbol{\theta }}(x)$.
%
%
We assume that the null hypothesis is true, and $\boldsymbol{\theta }_{0} \in \Theta_0$ is the true value of the parameter. Since the model is correct $\frac{\partial }{\partial \boldsymbol{\theta }^{T}}\frac{%
\partial }{\partial \boldsymbol{\theta }}\left. h_{n}(\boldsymbol{\theta }%
)\right\vert _{\boldsymbol{\theta =\theta }_{0}}$ converges in probability
to 
\begin{align}
\lim_{n\rightarrow \infty }\frac{\partial }{\partial \boldsymbol{\theta }^{T}%
}\frac{\partial }{\partial \boldsymbol{\theta }}h_{n}(\boldsymbol{\theta }) \Big\vert _{\boldsymbol{\theta =\theta }_{0}}&
=(1+\beta )\int \boldsymbol{u}_{\boldsymbol{\theta }_{0}}(x)\boldsymbol{u}_{%
\boldsymbol{\theta }_{0}}^{T}(x)f_{\boldsymbol{\theta }_{0}}^{1+\beta
}(x)dx-\int \boldsymbol{I}_{\boldsymbol{\theta }_{0}}(x)f_{%
\boldsymbol{\theta }_{0}}^{1+\beta }(x)dx  \notag \\
& -\beta \int \boldsymbol{u}_{\boldsymbol{\theta }_{0}}(x)\boldsymbol{u}_{%
\boldsymbol{\theta }_{0}}^{T}(x)f_{\boldsymbol{\theta }_{0}}^{1+\beta
}(x)dx+\int \boldsymbol{I}_{\boldsymbol{\theta }_{0}}(x)f_{%
\boldsymbol{\theta }_{0}}^{1+\beta }(x)dx  \notag \\
& =\int \boldsymbol{u}_{\boldsymbol{\theta }_{0}}(x)\boldsymbol{u}_{%
\boldsymbol{\theta }_{0}}^{T}(x)f_{\boldsymbol{\theta }_{0}}^{1+\beta }(x)dx.
\notag
\end{align}%
%
%
Notice that $\lim_{n\rightarrow \infty }\frac{\partial }{\partial
\boldsymbol{\theta }^{T}}\frac{\partial }{\partial \boldsymbol{\theta }}%
h_{n}(\boldsymbol{\theta })\vert _{\boldsymbol{\theta =\theta }_{0}}=\boldsymbol{J}_\beta(\boldsymbol{\theta }_{0})$
defined earlier in equation (\ref{model_variance}).
Since $f_{\boldsymbol{\theta }_{0}}$ represents the true distribution, some
simple algebra establishes that
\begin{equation*}
{\rm E}\left[n^{1/2}\frac{\partial }{\partial \boldsymbol{\theta }}\left. h_{n}(%
\boldsymbol{\theta })\right\vert _{\boldsymbol{\theta =\theta }_{0}}\right]=%
\boldsymbol{0}_{p},~~\mathrm{and ~~ Var} \left[n^{1/2}\frac{\partial }{\partial
\boldsymbol{\theta }}\left. h_{n}(\boldsymbol{\theta })\right\vert _{%
\boldsymbol{\theta =\theta }_{0}}\right]=\boldsymbol{K}_\beta(\boldsymbol{\theta }_{0})
\end{equation*}%
%
%
where $\boldsymbol{K}_\beta(\boldsymbol{\theta }_{0})$ is as
defined in equation (\ref{m_v}). Thus, asymptotically, $n^{1/2}\tfrac{%
\partial }{\partial \boldsymbol{\theta }}\left. h_{n}(\boldsymbol{\theta }%
)\right\vert _{\boldsymbol{\theta =\theta }_{0}}$ has a $\mathcal{N}(%
\boldsymbol{0}_{p},\boldsymbol{K}_\beta(\boldsymbol{\theta }_{0}))$ distribution.

The restricted minimum density power divergence estimator of $\boldsymbol{%
\theta }$, i.e. $\widetilde{\boldsymbol{\theta }}_{\beta }$, will satisfy
\begin{equation}
\left\{
\begin{array}{r}
n\tfrac{\partial }{\partial \boldsymbol{\theta }}\left. h_{n}(\boldsymbol{%
\theta })\right\vert _{\boldsymbol{\theta =}\widetilde{\boldsymbol{\theta }}%
_{\beta }}+\boldsymbol{G}(\widetilde{\boldsymbol{\theta }}_{\beta })%
\boldsymbol{\lambda }_{n}=\boldsymbol{0}_{p}, \\
\boldsymbol{g}(\widetilde{\boldsymbol{\theta }}_{\beta })=\boldsymbol{0}_{r},%
\end{array}%
\right.  \label{EQ:lagrange_restrictions}
\end{equation}%
%
%
%
%
where $\boldsymbol{\lambda }_{n}$ is a vector of Lagrangian multipliers. 
Now
we consider the Taylor expansion of $\tfrac{\partial }{\partial \boldsymbol{%
\theta }}\left. h_{n}(\boldsymbol{\theta })\right\vert _{\boldsymbol{\theta =%
}\widetilde{\boldsymbol{\theta }}_{\beta }}$\ about the point $\boldsymbol{%
\theta }_{0}$%
\begin{equation}
\tfrac{\partial }{\partial \boldsymbol{\theta }}\left. h_{n}(\boldsymbol{%
\theta })\right\vert _{\boldsymbol{\theta =}\widetilde{\boldsymbol{\theta }}%
_{\beta }}=\tfrac{\partial }{\partial \boldsymbol{\theta }}\left. h_{n}(%
\boldsymbol{\theta })\right\vert _{\boldsymbol{\theta =\theta }_{0}}+\frac{%
\partial }{\partial \boldsymbol{\theta }^{T}}\frac{\partial }{\partial
\boldsymbol{\theta }}\left. h_{n}(\boldsymbol{\theta })\right\vert _{%
\boldsymbol{\theta =\theta }_{1}}(\widetilde{\boldsymbol{\theta }}_{\beta }-%
\boldsymbol{\theta }_0),
\label{delh}
\end{equation}%
where $\boldsymbol{\theta }_1$ belongs to the line segment joining $\boldsymbol{\theta }_0$ and $\widetilde{\boldsymbol{\theta }}_{\beta }$.
%
%
%
%
%
%
%
Now using the Khintchine's weak law of large numbers we have
\begin{equation*}
 \frac{%
\partial }{\partial \boldsymbol{\theta }}\frac{\partial }{\partial
\boldsymbol{\theta }^{T}}\left. h_{n}(\boldsymbol{\theta })\right\vert _{%
\boldsymbol{\theta =\theta }_{1}} \underset%
{n\rightarrow \infty }{\overset{\mathcal{P}}{\longrightarrow }} 
\boldsymbol{J}_\beta(\boldsymbol{\theta }%
_{0}).
\end{equation*}
Therefore, from (\ref{delh}) we get
\begin{equation}
n^{1/2}\tfrac{\partial }{\partial \boldsymbol{\theta }}\left. h_{n}(%
\boldsymbol{\theta })\right\vert _{\boldsymbol{\theta =}\widetilde{%
\boldsymbol{\theta }}_{\beta }}=n^{1/2}\tfrac{\partial }{\partial
\boldsymbol{\theta }}\left. h_{n}(\boldsymbol{\theta })\right\vert _{%
\boldsymbol{\theta =\theta }_{0}}+\boldsymbol{J}_\beta(\boldsymbol{\theta }%
_{0})n^{1/2}(\widetilde{\boldsymbol{\theta }}_{\beta }-\boldsymbol{\theta }%
_{0})+o_{p}(1).  \label{EQ:Theorem2_B}
\end{equation}%
On the other hand, the Taylor expansion of $\boldsymbol{g}(\widetilde{%
\boldsymbol{\theta }}_{\beta })$\ about the point $\boldsymbol{\theta }_{0}$
is%
\begin{equation}
n^{1/2}\boldsymbol{g}(\widetilde{\boldsymbol{\theta }}_{\beta })=\boldsymbol{%
G}^{T}(\boldsymbol{\theta }_{0})n^{1/2}(\widetilde{\boldsymbol{\theta }}%
_{\beta }-\boldsymbol{\theta }_{0})+o_{p}(1).  \label{EQ:Theorem2_C}
\end{equation}%
%
%
%
Combining equations (\ref{EQ:lagrange_restrictions}) and (\ref{EQ:Theorem2_B}) we have
\begin{equation}
n^{1/2}\tfrac{\partial }{\partial \boldsymbol{\theta }}\left. h_{n}(%
\boldsymbol{\theta })\right\vert _{\boldsymbol{\theta =\theta }_{0}}+%
\boldsymbol{J}_\beta(\boldsymbol{\theta }_{0})n^{1/2}(\widetilde{\boldsymbol{%
\theta }}_{\beta }-\boldsymbol{\theta }_{0})+\boldsymbol{G}(\boldsymbol{%
\theta }_{0})n^{-1/2}\boldsymbol{\lambda }_{n}+o_{p}(1)=\boldsymbol{0}_{p}.
\label{EQ:Theorem2_D}
\end{equation}%
%
%
%
The last expression also uses the fact that $\boldsymbol{G}(\widetilde{\boldsymbol{\theta }}_{\beta }) - \boldsymbol{G}(\boldsymbol{%
\theta }_{0})$ is an $o_p(1)$ term. Similarly from (\ref{EQ:lagrange_restrictions}) and (\ref{EQ:Theorem2_C}) it follows that
\begin{equation}
\boldsymbol{G}^{T}(\boldsymbol{\theta }_{0})n^{1/2}(\widetilde{\boldsymbol{%
\theta }}_{\beta }-\boldsymbol{\theta }_{0})+o_p(1)=\boldsymbol{0}%
_{r}.  \label{EQ:Theorem2_E}
\end{equation}%
%
%
%
%
%
%
Now we can express equations (\ref{EQ:Theorem2_D}) and (\ref{EQ:Theorem2_E})
in the matrix form as
\begin{equation*}
\left(
\begin{array}{cc}
\boldsymbol{J}_\beta(\boldsymbol{\theta }_{0}) & \boldsymbol{G}(\boldsymbol{\theta
}_{0}) \\
\boldsymbol{G}^{T}(\boldsymbol{\theta }_{0}) & \boldsymbol{0}_{r\times r}%
\end{array}%
\right) \left(
\begin{array}{c}
n^{1/2}(\widetilde{\boldsymbol{\theta }}_{\beta }-\boldsymbol{\theta }_{0})
\\
n^{-1/2}\boldsymbol{\lambda }_{n}%
\end{array}%
\right) =\left(
\begin{array}{c}
-n^{1/2}\frac{\partial }{\partial \boldsymbol{\theta }}\left. h_{n}(%
\boldsymbol{\theta })\right\vert _{\boldsymbol{\theta =\theta }_{0}} \\
\boldsymbol{0}_{r}%
\end{array}%
\right) +o_{p}(1).
\end{equation*}%
%
%
%
%
%
Therefore
\begin{equation*}
\left(
\begin{array}{c}
n^{1/2}(\widetilde{\boldsymbol{\theta }}_{\beta }-\boldsymbol{\theta }_{0})
\\
n^{-1/2}\boldsymbol{\lambda }_{n}%
\end{array}%
\right) =\left(
\begin{array}{cc}
\boldsymbol{J}_\beta(\boldsymbol{\theta }_{0}) & \boldsymbol{G}(\boldsymbol{\theta
}_{0}) \\
\boldsymbol{G}^{T}(\boldsymbol{\theta }_{0}) & \boldsymbol{0}_{r\times r}%
\end{array}%
\right) ^{-1}\left(
\begin{array}{c}
-n^{1/2}\frac{\partial }{\partial \boldsymbol{\theta }}\left. h_{n}(%
\boldsymbol{\theta })\right\vert _{\boldsymbol{\theta =\theta }_{0}}\\
\boldsymbol{0}_{r}%
\end{array}%
\right) +o_{p}(1) .
\end{equation*}%
%
%
%
%
But
\begin{equation*}
\left(
\begin{array}{cc}
\boldsymbol{J}_\beta(\boldsymbol{\theta }_{0}) & \boldsymbol{G}(\boldsymbol{\theta
}_{0}) \\
\boldsymbol{G}^{T}(\boldsymbol{\theta }_{0}) & \boldsymbol{0}%
\end{array}%
\right) ^{-1}={\left(
\begin{array}{cc}
\boldsymbol{P}(\boldsymbol{\theta }_{0}) & \boldsymbol{Q}(\boldsymbol{\theta
}_{0}) \\
\boldsymbol{Q}(\boldsymbol{\theta }_{0})^{T} & \boldsymbol{R}(\boldsymbol{%
\theta }_{0})%
\end{array}%
\right) } ,
\end{equation*}%
%
%
%
where
$\boldsymbol{P}(\boldsymbol{\theta }_{0})$ and $\boldsymbol{Q}(\boldsymbol{\theta }_{0})$ are as given in (\ref{P}) and (\ref{Q}) respectively.
The matrix $\boldsymbol{R}(\boldsymbol{\theta}_0)$ is the quantity needed to make the right hand side of the above equation equal to the indicated inverse. 
 Then
\begin{equation}
n^{1/2}(\widetilde{\boldsymbol{\theta }}_{\beta }-\boldsymbol{\theta }_{0})=-%
\boldsymbol{P}(\boldsymbol{\theta }_{0})n^{1/2}\frac{\partial }{\partial
\boldsymbol{\theta }}\left. h_{n}(\boldsymbol{\theta })\right\vert _{%
\boldsymbol{\theta =\theta }_{0}}+o_{p}(1) , \label{[E]}
\end{equation}%
%
%
%
and we know
%
\begin{equation}
n^{1/2}\frac{\partial }{\partial \boldsymbol{\theta }}\left. h_{n}(%
\boldsymbol{\theta })\right\vert _{\boldsymbol{\theta =\theta }_{0}}\underset%
{n\rightarrow \infty }{\overset{\mathcal{L}}{\longrightarrow }}\mathcal{N%
}(\boldsymbol{0},\boldsymbol{K}_{\beta }(\boldsymbol{\theta }_0)).  \label{[F]}
\end{equation}%
%
%
%
%
Finally combining (\ref{[E]}) and (\ref{[F]}) we get the desired result. 

\bigskip
\noindent \textbf{Proof of Theorem \ref{Theorem3} } 
Consider the expression $d_{\gamma }(f_{\boldsymbol{\theta }},f_{\widetilde{%
\boldsymbol{\theta }}_{\beta }})$. A Taylor expansion for an arbitrary $%
\boldsymbol{\theta }\in \Theta $, around $\widetilde{\boldsymbol{\theta }}%
_{\beta }$ leads to the relation
\begin{align*}
d_{\gamma }(f_{\boldsymbol{\theta }},f_{\widetilde{\boldsymbol{\theta }}%
_{\beta }})& =d_{\gamma }(f_{\widetilde{\boldsymbol{\theta }}_{\beta }},f_{%
\widetilde{\boldsymbol{\theta }}_{\beta }})+{\textstyle\sum\limits_{i=1}^{p}}%
\left( \frac{\partial d_{\gamma }(f_{\boldsymbol{\theta }},f_{\widetilde{%
\boldsymbol{\theta }}_{\beta }})}{\partial \theta _{i}}\right) _{\boldsymbol{%
\theta =}\widetilde{\boldsymbol{\theta }}_{\beta }}\left( \theta _{i}-%
\widetilde{\theta }_{i,{\beta }}\right) \\
& +\frac{1}{2}{\textstyle\sum\limits_{i=1}^{p}}{\textstyle%
\sum\limits_{j=1}^{p}}\left( \frac{\partial ^{2}d_{\gamma }(f_{\boldsymbol{%
\theta }},f_{\widetilde{\boldsymbol{\theta }}_{\beta }})}{\partial \theta
_{i}\partial \theta _{j}}\right) _{\boldsymbol{\theta =}\widetilde{%
\boldsymbol{\theta }}_{\beta }}\left( \theta _{i}-\widetilde{\theta }_{i,{%
\beta }}\right) \left( \theta _{j}-\widetilde{\theta }_{j,{\beta }}\right)
+o\left( \left\Vert \boldsymbol{\theta }-\widetilde{\boldsymbol{\theta }}%
_{\beta }\right\Vert ^{2}\right) .
\end{align*}%
%
%
%
%
%
It is clear that $d_{\gamma }(f_{\widetilde{\boldsymbol{\theta }}_{\beta
}},f_{\widetilde{\boldsymbol{\theta }}_{\beta }})=0$, $\left( \frac{\partial
d_{\gamma }(f_{\boldsymbol{\theta }},f_{\widetilde{\boldsymbol{\theta }}%
_{\beta }})}{\partial \boldsymbol{\theta }_{i}}\right) _{\boldsymbol{\theta =%
}\widetilde{\boldsymbol{\theta }}_{\beta }}=0$ for each $i$, and
\begin{equation*}
a_{ij}^{\gamma }\left( \widetilde{\boldsymbol{\theta }}_{\beta }\right)
=\left( \frac{\partial ^{2}d_{\gamma }(f_{\boldsymbol{\theta }},f_{%
\widetilde{\boldsymbol{\theta }}_{\beta }})}{\partial \theta _{i}\partial
\theta _{j}}\right) _{\boldsymbol{\theta =}\widetilde{\boldsymbol{\theta }}%
_{\beta }}=\left( 1+\gamma \right) \int\nolimits_{\mathcal{X}}f_{\widetilde{%
\boldsymbol{\theta }}_{\beta }}^{\gamma -1}\left( x\right) \frac{\partial f_{%
\widetilde{\boldsymbol{\theta }}_{\beta }}\left( x\right) }{\partial \theta
_{i}}\frac{\partial f_{\widetilde{\boldsymbol{\theta }}_{\beta }}\left(
x\right) }{\partial \theta _{j}}dx.
\end{equation*}%
%
%
%
%
%
%
Therefore,
\begin{equation*}
T_{\boldsymbol{\gamma }}(\widehat{\boldsymbol{\theta }}_{\beta },\widetilde{%
\boldsymbol{\theta }}_{\beta })=2nd_{\gamma }(f_{\widehat{\boldsymbol{\theta
}}},f_{\widetilde{\boldsymbol{\theta }}_{\beta }})=n^{1/2}(\widehat{%
\boldsymbol{\theta }}_{\beta }-\widetilde{\boldsymbol{\theta }}_{\beta })^{T}%
\boldsymbol{A}_{\gamma }(\widetilde{\boldsymbol{\theta }}_{\beta })n^{1/2}(%
\widehat{\boldsymbol{\theta }}_{\beta }-\widetilde{\boldsymbol{\theta }}%
_{\beta })+n\times o\left( \left\Vert \widehat{\boldsymbol{\theta }}_{\beta
}-\widetilde{\boldsymbol{\theta }}_{\beta })\right\Vert ^{2}\right) .
\end{equation*}%
%
%
%
%
%
%
%
%
%
%
%
%
%
%
%
%
%
%
%
%
Under $\boldsymbol{\theta }_{0} \in \Theta_0$
\begin{equation*}
\boldsymbol{A}_{\gamma }(\widetilde{\boldsymbol{\theta }}_{\beta })\underset{%
n\rightarrow \infty }{\overset{\mathcal{P}}{\longrightarrow }}%
\boldsymbol{A}_{\gamma }\left( \boldsymbol{\theta }_{0}\right) .
\end{equation*}%
%
%
%
%
%
Using 
(\ref{[E]}) and
\begin{equation*}
n^{1/2}\tfrac{\partial }{\partial \boldsymbol{\theta }}\left. h_{n}(%
\boldsymbol{\theta })\right\vert _{\boldsymbol{\theta =\theta }_{0}}=-n^{1/2}%
\boldsymbol{J}_\beta\left( \boldsymbol{\theta }_{0}\right) (\widehat{\boldsymbol{%
\theta }}_{\beta }-\boldsymbol{\theta }_{0})+o_{p}(1),
\end{equation*}%
%
%
%
%
we get
\begin{align*}
n^{1/2}(\widetilde{\boldsymbol{\theta }}_{\beta }-\boldsymbol{\theta }_{0})&
=\boldsymbol{P}(\boldsymbol{\theta }_{0})n^{1/2}\boldsymbol{J}_\beta\left(
\boldsymbol{\theta }_{0}\right) (\widehat{\boldsymbol{\theta }}_{\beta }-%
\boldsymbol{\theta }_{0})+o_{p}(1) \\
& =\boldsymbol{J}_\beta^{-1}\left( \boldsymbol{\theta }_{0}\right) n^{1/2}%
\boldsymbol{J}_\beta\left( \boldsymbol{\theta }_{0}\right) (\widehat{\boldsymbol{%
\theta }}_{\beta }-\boldsymbol{\theta }_{0})-\boldsymbol{Q}\left(
\boldsymbol{\theta }_{0}\right) \boldsymbol{G}^{T}(\boldsymbol{\theta }_{0}) n^{1/2}(\widehat{\boldsymbol{\theta }}%
_{\beta }-\boldsymbol{\theta }_{0})+o_{p}(1) \\
& =n^{1/2}(\widehat{\boldsymbol{\theta }}_{\beta }-\boldsymbol{\theta }_{0})-%
\boldsymbol{Q}\left( \boldsymbol{\theta }_{0}\right) \boldsymbol{G}^{T}(\boldsymbol{\theta }_{0}) n^{1/2}(\widehat{%
\boldsymbol{\theta }}_{\beta }-\boldsymbol{\theta }_{0})+o_{p}(1).
\end{align*}%
%
%
%
%
Therefore
\begin{equation}
n^{1/2}(\widehat{\boldsymbol{\theta }}_{\beta }-\widetilde{\boldsymbol{%
\theta }}_{\beta })=\boldsymbol{Q}(\boldsymbol{\theta }_{0}) \boldsymbol{G}^{T}(\boldsymbol{\theta }_{0}) n^{1/2}(\widehat{%
\boldsymbol{\theta }}_{\beta }-\boldsymbol{\theta }_{0})+o_{p}(1).  
\label{curl_theta}
\end{equation}%
%
%
%
On the other hand, $n^{1/2}(\widehat{\boldsymbol{\theta }}_{\beta }-%
\boldsymbol{\theta }_{0})\underset{n\rightarrow \infty }{\overset{%
\mathcal{L}}{\longrightarrow }}\mathcal{N}(0,\boldsymbol{J}_{\beta }^{-1}(%
\boldsymbol{\theta }_{0})\boldsymbol{K}_\beta(\boldsymbol{\theta }_{0})\boldsymbol{%
J}_{\beta }^{-1}(\boldsymbol{\theta }_{0}))$. From equations (\ref{Q}) and (\ref{B}) we have $\boldsymbol{B}_{\beta } 
\left( \boldsymbol{\theta }_{0}\right) = \boldsymbol{Q}_{\beta } \left( \boldsymbol{%
\theta }_{0}\right) \boldsymbol{G}^{T}(\boldsymbol{\theta }_{0}) \boldsymbol{J}_{\beta }^{-1}(%
\boldsymbol{\theta }_{0})$. Therefore it follows that
\begin{equation*}
n^{1/2}(\widehat{\boldsymbol{\theta }}_{\beta }-\widetilde{\boldsymbol{%
\theta }}_{\beta })\underset{n\rightarrow \infty }{\overset{\mathcal{L}}{%
\longrightarrow }}\mathcal{N}(0,\boldsymbol{B}_{\beta }\left( \boldsymbol{%
\theta }_{0}\right) \boldsymbol{K}_{\beta }(\boldsymbol{\theta }_{0})%
\boldsymbol{B}_{\beta }\left( \boldsymbol{\theta }_{0}\right) ).
\end{equation*}%
Now the asymptotic distribution of the random variables
$T_{\boldsymbol{\gamma }}(\widehat{\boldsymbol{\theta }}_{\beta },\widetilde{%
\boldsymbol{\theta }}_{\beta })=2nd_{\gamma }(f_{\widehat{\boldsymbol{\theta
}}_{\beta }},f_{\widetilde{\boldsymbol{\theta }}_{\beta }})$
and
\begin{equation*}
n^{1/2}(\widehat{\boldsymbol{\theta }}_{\beta }-\widetilde{\boldsymbol{%
\theta }}_{\beta })^{T}\boldsymbol{A}_{\gamma }\left( \boldsymbol{\theta }%
_{0}\right) n^{1/2}(\widehat{\boldsymbol{\theta }}_{\beta }-\widetilde{%
\boldsymbol{\theta }}_{\beta })
\end{equation*}%
%
%
%
%
%
are the same because
\begin{equation*}
n\times o\left( \left\Vert \widehat{\boldsymbol{\theta }}_{\beta }-%
\widetilde{\boldsymbol{\theta }}_{\beta }\right\Vert ^{2}\right)
=o_{p}\left( 1\right) .
\end{equation*}%
%
%
%
%
Now we apply Corollary 2.1 in  \cite{MR801686}, which essentially states the following.
 Let $\boldsymbol{X}$ be a $q$-variate normal random
variable with mean vector $\boldsymbol{0}$ and variance-covariance matrix $%
\boldsymbol{\Sigma }$. Let $\boldsymbol{M}$ be a real symmetric matrix of
order $q$. Let $k=\rm{rank}(\boldsymbol{\Sigma M\Sigma )}$, $k\geq 1$ and let $%
\lambda _{1},\ldots ,\lambda _{k},$ be the nonzero eigenvalues of $%
\boldsymbol{M\Sigma .}$ Then the distribution of the quadratic form $%
\boldsymbol{X}^{T}\boldsymbol{MX}$ coincides with the distribution of the
random variable ${\textstyle\sum\limits_{i=1}^{k}}\lambda _{i}Z_{i}^{2},$
where $Z_{1},\ldots ,Z_{k}$ are independent, each being a standard normal variable. In our case the asymptotic distribution of
$T_{\gamma }(\widehat{\boldsymbol{\theta }}_{\beta },\widetilde{\boldsymbol{%
\theta }}_{\beta })$
coincides with the distribution of the random variable
${\textstyle\sum\limits_{i=1}^{k}}\lambda _{i}^{\beta ,\gamma }Z_{i}^{2}$
where $\lambda _{1}^{\beta ,\gamma },\ldots ,\lambda _{k}^{\beta ,\gamma }$,
are the nonzero eigenvalues of $\boldsymbol{A}_{\gamma }\left( \boldsymbol{%
\theta }_{0}\right) \boldsymbol{B}_{\beta }\left( \boldsymbol{\theta }%
_{0}\right) \boldsymbol{K}_{\beta }(\boldsymbol{\theta }_{0})\boldsymbol{B}%
_{\beta }\left( \boldsymbol{\theta }_{0}\right) $,
where
\begin{equation*}
k={\rm rank}\left( \boldsymbol{B}_{\beta }\left( \boldsymbol{\theta }_{0}\right)
\boldsymbol{K}_{\beta }(\boldsymbol{\theta }_{0})\boldsymbol{B}_{\beta
}\left( \boldsymbol{\theta }_{0}\right) \boldsymbol{A}_{\gamma }\left(
\boldsymbol{\theta }_{0}\right) \boldsymbol{B}_{\beta }\left( \boldsymbol{%
\theta }_{0}\right) \boldsymbol{K}_{\beta }(\boldsymbol{\theta }_{0})%
\boldsymbol{B}_{\beta }\left( \boldsymbol{\theta }_{0}\right) \right) .%
\end{equation*}%
%
%

\bigskip
\noindent \textbf{Proof of Theorem \ref{theorem9} } 
Notice that $\boldsymbol{Q}_{\beta } \left( \boldsymbol{%
\theta }_{0}\right) \boldsymbol{G}^{T}(\boldsymbol{\theta }_{0}) = \boldsymbol{B}_{\beta } 
\left( \boldsymbol{\theta }_{0}\right)  \boldsymbol{J}_{\beta }(%
\boldsymbol{\theta }_{0})$. From equation (\ref{curl_theta}) we have
\begin{equation*}
n^{1/2}(\widehat{\boldsymbol{\theta }}_{\beta }-\widetilde{\boldsymbol{%
\theta }}_{\beta })=\boldsymbol{B}_{\beta }\left( \boldsymbol{\theta }%
_{0}\right) \boldsymbol{\boldsymbol{J}_{\beta }(\boldsymbol{\theta }_{0})}%
n^{1/2}(\widehat{\boldsymbol{\theta }}_{\beta }-\boldsymbol{\theta }%
_{0})+o_{p}(1),
\end{equation*}%
%
%
%
%
%
%
%
%
%
%
then
\begin{align*}
n^{1/2}(\widehat{\boldsymbol{\theta }}_{\beta }-\widetilde{\boldsymbol{%
\theta }}_{\beta })& =\boldsymbol{B}_{\beta }\left( \boldsymbol{\theta }%
_{0}\right) \boldsymbol{\boldsymbol{J}_{\beta }(\boldsymbol{\theta }_{0})}%
n^{1/2}(\widehat{\boldsymbol{\theta }}_{\beta }-\boldsymbol{\theta }_{n})+%
\boldsymbol{B}_{\beta }\left( \boldsymbol{\theta }_{0}\right) \boldsymbol{%
\boldsymbol{J}_{\beta }(\boldsymbol{\theta }_{0})}n^{1/2}\left( \boldsymbol{%
\ \theta }_{n}-\boldsymbol{\theta }_{0}\right) +o_{p}(1) \\
& =\boldsymbol{B}_{\beta }\left( \boldsymbol{\theta }_{0}\right) \boldsymbol{%
\boldsymbol{J}_{\beta }(\boldsymbol{\theta }_{0})}n^{1/2}(\widehat{%
\boldsymbol{\theta }}_{\beta }-\boldsymbol{\theta }_{n})+\boldsymbol{B}%
_{\beta }\left( \boldsymbol{\theta }_{0}\right) \boldsymbol{\boldsymbol{J}%
_{\beta }(\boldsymbol{\theta }_{0})d}+o_{p}(1).
\end{align*}%
%
%
%
%
%
%
%
%
%
%
Under $H_{1,n}$ one has
\begin{equation*}
n^{1/2}(\widehat{\boldsymbol{\theta }}_{\beta }-\boldsymbol{\theta }_{n})%
\overset{\mathcal{L}}{\underset{n\rightarrow \infty }{\longrightarrow }}%
\mathcal{N}\left( \mathbf{0}_{p},\boldsymbol{J}_{\beta }^{-1}(\boldsymbol{%
\theta }_{0})\boldsymbol{K}_{\beta }(\boldsymbol{\theta }_{0})\boldsymbol{J}%
_{\beta }^{-1}(\boldsymbol{\theta }_{0})\right),
\end{equation*}%
%
%
%
%
and
\begin{equation*}
n^{1/2}(\widehat{\boldsymbol{\theta }}_{\beta }-\widetilde{\boldsymbol{%
\theta }}_{\beta })\overset{\mathcal{L}}{\underset{n\rightarrow \infty }{%
\longrightarrow }}\mathcal{N}\left( \boldsymbol{B}_{\beta }\left(
\boldsymbol{\theta }_{0}\right) \boldsymbol{\boldsymbol{J}_{\beta }(%
\boldsymbol{\theta }_{0})d},\boldsymbol{B}_{\beta }\left( \boldsymbol{\theta
}_{0}\right) \boldsymbol{K}_{\beta }(\boldsymbol{\theta }_{0})\boldsymbol{B}%
_{\beta }\left( \boldsymbol{\theta }_{0}\right) \right) .
\end{equation*}%
%
%
%
%
We know that
\begin{equation*}
T_{\gamma }(\widehat{\boldsymbol{\theta }}_{\beta },\widetilde{\boldsymbol{%
\theta }}_{\beta })=n^{1/2}(\widehat{\boldsymbol{\theta }}_{\beta }-%
\widetilde{\boldsymbol{\theta }}_{\beta })^{T}\boldsymbol{A}_{\gamma }(%
\boldsymbol{\theta }_{0})n^{1/2}(\widehat{\boldsymbol{\theta }}_{\beta }-%
\widetilde{\boldsymbol{\theta }}_{\beta })+o_{p}(1).
\end{equation*}%
Then, $T_{\gamma }(\widehat{\boldsymbol{\theta }}_{\beta },\widetilde{%
\boldsymbol{\theta }}_{\beta })$ has the same asymptotic distribution as the
quadratic form $n^{1/2}(\widehat{\boldsymbol{\theta }}_{\beta }-\widetilde{%
\boldsymbol{\theta }}_{\beta })^{T}\boldsymbol{A}_{\gamma }(\boldsymbol{%
\theta }_{0})n^{1/2}(\widehat{\boldsymbol{\theta }}_{\beta }-\widetilde{%
\boldsymbol{\theta }}_{\beta }).$ Now the result follows from Corollary 2.2
of  \cite{MR801686}: Let $\boldsymbol{X}\sim \mathcal{N}_{q}(\boldsymbol{%
\mu },\boldsymbol{\Sigma })$, a $q$-variate normal distribution. Let $%
\boldsymbol{M}$ be a real symmetric non-negative definite matrix of order $q$%
. Let $k=\rm{rank}(\boldsymbol{\Sigma M\Sigma })$, $k\geq 1$, and let $\lambda
_{1},\ldots ,\lambda _{k}$ be the positive eigenvalues of $\boldsymbol{%
M\Sigma }$. Then the quadratic form $\boldsymbol{X}^{T}\boldsymbol{MX}$ has
the same distribution as the random variable%
\begin{equation*}
\sum\limits_{i=1}^{k}\lambda _{i}\left( Z_{i}+w_{i}\right) ^{2}+\eta ,
\end{equation*}%
%
%
%
%
%
%
%
%
%
%
where $Z_{1},\ldots ,Z_{k}$ are independent, each having a standard normal
distribution. Values of $\boldsymbol{w}$ and $\eta $ are given by
\begin{equation*}
\boldsymbol{w=\Lambda }_{k}^{-1}\boldsymbol{V}^{T}\boldsymbol{S}^{T}%
\boldsymbol{M\mu },~~\eta =\boldsymbol{\mu }^{T}\boldsymbol{%
M\mu }-\boldsymbol{w}^{T}\boldsymbol{\Lambda }_{k}\boldsymbol{w},
\end{equation*}%
%
%
%
%
%
%
%
%
%
%
where $\boldsymbol{S}$ is any $q\times k$ square root of $\boldsymbol{\Sigma
}$, $\boldsymbol{\Lambda }_{k}= \rm{diag}\left( \lambda _{1},\ldots ,\lambda
_{k}\right) $ and $\boldsymbol{V}$ is the matrix of corresponding
orthonormal eigenvectors. We therefore have the desired result.

\end{document}